\documentclass[11pt]{article}
\usepackage[utf8]{inputenc}
\usepackage[]{amsmath}
\usepackage{amsthm}
\usepackage[]{upgreek}
\usepackage[]{accents}
\usepackage{xcolor}
\usepackage{mathtools}
\usepackage[]{gensymb}
\usepackage{changepage}
\usepackage[]{graphicx}
\usepackage[]{nicefrac}
\usepackage{tikz-cd} 
\usepackage{pgfplots}
\usetikzlibrary{decorations.pathmorphing}
\usetikzlibrary{decorations.markings}
\usetikzlibrary{patterns}
\usetikzlibrary{shapes}
\usetikzlibrary{plotmarks}
\usepackage[]{textgreek}
\usepackage[]{url}
\usepackage{xcolor,cancel}
\usepackage{pdfpages}
\usepackage[inline]{enumitem}
\usepackage[]{float}
\usepackage{braket}
\usepackage{graphicx, amssymb}
\usepackage{mathrsfs}
\usepackage{geometry}
\usepackage{fancyhdr}
\usepackage{sectsty}
\usepackage{slashed}
\usepackage[toc,page]{appendix}
\usepackage{tocloft}
\usepackage{lscape}
\usepackage{bbm}
\usepackage{tcolorbox}
\usepackage{multicol}
\usepackage[numbers]{natbib}
\usepackage[hidelinks]{hyperref}

\newtheorem{theorem}{Theorem}[section]
\newtheorem{definition}[theorem]{Definition}

\newtheorem{corollary}[theorem]{Corollary}
\newtheorem{prop}[theorem]{Proposition}
\newtheorem{lemma}[theorem]{Lemma}
\newtheorem{remark}[theorem]{Remark}

\newtheorem*{conjecture*}{Conjecture}
\newtheorem*{theorem*}{Theorem}
\newtheorem*{corollary*}{Corollary}
\newtheorem*{Question*}{Question}

\renewcommand{\phi}{\varphi}

\newcommand{\ns}{\slashed{\nabla}}
\newcommand{\divs}{\slashed{\mathrm{div}}}

\newcommand{\curls}{\slashed{\mathrm{curl}}}
\newcommand{\Dts}{\slashed{\mathcal{D}}_2^{\star}}

\newcommand{\Olin}{\Omega^{-1}\accentset{\scalebox{.6}{\mbox{\tiny (1)}}}{\Omega}}
\newcommand{\Olino}{\accentset{\scalebox{.6}{\mbox{\tiny (1)}}}{\Omega}}

\newcommand{\glin}{\accentset{\scalebox{.6}{\mbox{\tiny (1)}}}{\slashed{g}}}
  
\newcommand{\bmlin}{\accentset{\scalebox{.6}{\mbox{\tiny (1)}}}{b}}

\newcommand{\xblin}{\accentset{\scalebox{.6}{\mbox{\tiny (1)}}}{\underline{\hat{\chi}}}}
\newcommand{\xlin}{\accentset{\scalebox{.6}{\mbox{\tiny (1)}}}{{\hat{\chi}}}}

\newcommand{\eblin}{\accentset{\scalebox{.6}{\mbox{\tiny (1)}}}{\underline{\eta}}}
\newcommand{\elin}{\accentset{\scalebox{.6}{\mbox{\tiny (1)}}}{{\eta}}}
\newcommand{\otx}{\accentset{\scalebox{.6}{\mbox{\tiny (1)}}}{\left(\Omega \mathrm{tr} \chi\right)}}
\newcommand{\otxb}{\accentset{\scalebox{.6}{\mbox{\tiny (1)}}}{\left(\Omega \mathrm{tr} \underline{\chi}\right)}}
\newcommand{\olin}{\accentset{\scalebox{.6}{\mbox{\tiny (1)}}}{\omega}}
\newcommand{\olinb}{\accentset{\scalebox{.6}{\mbox{\tiny (1)}}}{\underline{\omega}}}

\newcommand{\ablin}{\accentset{\scalebox{.6}{\mbox{\tiny (1)}}}{\underline{\alpha}}}
\newcommand{\alin}{\accentset{\scalebox{.6}{\mbox{\tiny (1)}}}{{\alpha}}}
\newcommand{\bblin}{\accentset{\scalebox{.6}{\mbox{\tiny (1)}}}{\underline{\beta}}}
\newcommand{\blin}{\accentset{\scalebox{.6}{\mbox{\tiny (1)}}}{{\beta}}}
\newcommand{\rlin}{\accentset{\scalebox{.6}{\mbox{\tiny (1)}}}{\rho}}
\newcommand{\slin}{\accentset{\scalebox{.6}{\mbox{\tiny (1)}}}{{\sigma}}}

\newcommand{\Gamlin}{\accentset{\scalebox{.6}{\mbox{\tiny (1)}}}{\Gamma}}
\newcommand{\Wlin}{\accentset{\scalebox{.6}{\mbox{\tiny (1)}}}{\mathrm{W}}}

  \title{	\normalsize \textsc{} 	
		 		\LARGE \textbf{\uppercase{Coercivity Properties of the Canonical Energy in Double Null Gauge}}	
		}
		\author{
		Sam C. Collingbourne \footnote{scc@math.columbia.edu}\\
	Department of Mathematics, Columbia University, New York}

\begin{document}
\maketitle
\begin{abstract}
In this paper, we study the canonical energy associated with solutions to the linearised vacuum Einstein equation on a stationary spacetime. The main result of this paper establishes, in the context of the $4$-dimensional Schwarzschild exterior, a direct correspondence between the conservation law satisfied by the canonical energy and the conservation laws deduced by Holzegel for gravitational perturbations in double null gauge. Since the latter exhibit useful coercivity properties (leading to energy and pointwise boundedness statements) we obtain coercivity results for the canonical energy in the double null gauge as a corollary. More generally, the correspondence suggests a systematic way to uncover coercivity properties in the conservation laws for the canonical energy on Kerr.
\end{abstract}

\tableofcontents

\section{Introduction}

\subsection{The Schwarzschild Spacetime}
The $4$-dimensional Schwarzschild spacetime, $(\mathrm{Schw}_4,g_s)$, is the most basic solution to the vacuum Einstein equation,
\begin{align}\label{VE}
    \mathrm{Ric}[g]=0,
\end{align}
giving rise to the black hole phenomenon. One can cover the exterior region of the $4$-dimensional Schwarzschild spacetime with  double null Eddington--Finkelstein coordinates $(u,v,\theta,\phi)\in \mathbb{R}^2\times\mathbb{S}^2_{u,v}$  with
\begin{align*}
u=\frac{1}{2}(t-r_{\star}),\qquad v=\frac{1}{2}(t+r_{\star}),\qquad r_{\star}(r)\doteq (r-4M) + 2 M \ln\Big(\frac{r-2M}{2M}\Big)
\end{align*}
where standard~$(t,r,\theta,\phi)$ are the standard Schwarzschild coordinates. The metric on the exterior region in $(u,v,\theta,\phi)$ coordinates is
\begin{align}\label{SchTDN}
g_s=-2D(r(u,v))(du\otimes dv+dv\otimes du)+r(u,v)^2\gamma_{2},\qquad D(r(u,v))\doteq 1-\frac{2M}{r(u,v)},
\end{align}
where~$\gamma_{2}$ is the metric on the unit~$2$-sphere. The exterior region has the following Penrose diagrammatic represention:
\begin{figure}[H]
\centering
\begin{tikzpicture}

   \draw [thick] (6,3) -- (8.95,5.95);
       \draw [dashed] (9.05,5.95) -- (11.95,3.05);
       \draw [dashed] (9.05,0.05) -- (11.95,2.95);
    \draw [thick] (6,3) -- (8.95,0.05);   
      \node[mark size=2pt,inner sep=2pt] at (12,3) {$\circ$};
      \node[mark size=2pt] at (9,6) {$\circ$};
          \node[mark size=2pt] at (9,0) {$\circ$};

      \node (scriplus) at (10.9,4.75) {\large $\mathcal{I}^{+}$};
        \node (Hplus) at (7.4,4.75) {\large $\mathcal{H}^{+}$};   
        \node (scriplus) at (10.9,1.2) {\large $\mathcal{I}^{-}$};
        \node (Hplus) at (7.4,1.2) {\large $\mathcal{H}^{-}$};   
      \node (iplus) at (9.25,6.3) {\large $i^{+}$};
      \node (inaught) at (12.3,3) {\large $i^{0}$};

          
           \draw[blue] (7,3) -- (9,5);
    \draw[blue] (8,2) -- (10,4);
    \draw[blue] (7,3) -- (8,2);
     \draw [dotted,thick] (6.5,3.5) -- (7,3);
     \draw [dotted,thick] (9,5) -- (9.5,5.5);
     \draw [dotted,thick] (9,5) -- (8.5,5.5);
      \draw [dotted,thick] (10,4) -- (10.5,4.5);
    \draw[blue] (9,5) -- (10,4);

         \node at (7.85,4.3) {$C_{u_1}$};
            \node  at (9.15,2.7) {$C_{u_0}$};
            \node at (7.25,2.4) {$\underline{C}_{v_0}$};
            \node at (9.7,4.7) {$\underline{C}_{{v_1}}$}; 
            \node at (8.5,3.5) {$\mathcal{R}$};
      \end{tikzpicture}  
     \caption{The Penrose diagram of the exterior of $(\mathrm{Schw}_4,g_s)$ with a region $\mathcal{R}$ bounded by a characteristic rectangle.} \label{fig1}
\end{figure}
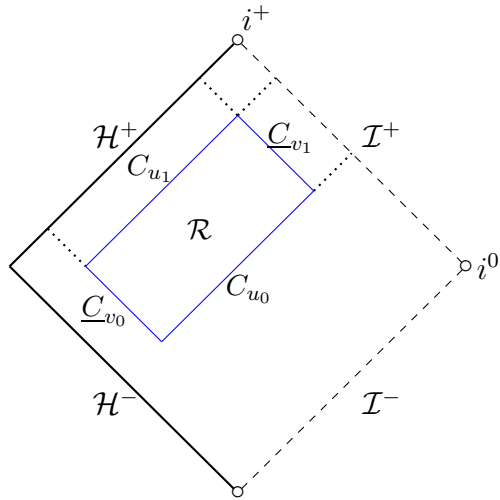
\noindent In these coordinates, constant $v$ or $u$ level sets are (ingoing or outgoing, respectively) null hypersurfaces denoted in figure~\ref{fig1} with $\underline{C}_v$ or $C_u$ respectively. Strictly speaking, the coordinates $(u,v,\theta,\phi)$ do not cover the future event horizon, $\mathcal{H}^+$, or future null infinity, $\mathcal{I}^+$. However, formally one can parameterise the future event horizon as $(\infty,v,\theta,\phi)$ and future null infinity as $(u,\infty,\theta,\phi)$.

\subsection{Conservation Laws, Canonical Energy and Coercivity}\label{sec:IntroCL}
This paper studies conservation laws for linearised gravity on the $4$-dimensional Schwarzschild black hole spacetime. Conservation laws play an important role in the study of stability properties of solutions to many partial differential equations. In particular, conservation laws are essential ingredients when trying to prove energy estimates for hyperbolic equations on black hole backgrounds. Such energy estimates then let one infer boundedness and decay properties of solutions to such equations (see, for example,~\cite{AndersonBlue,ABBSM19,DR,DHR2,DHR,DHRT,DR2,Mihalis2,DR4,DR3,Mihalis,DRY,Elena,GKS,Holzegel,Szeftel,Morawetz,YakovRita2}). In many cases, one can view such estimates as arising from applications of the divergence theorem to identities satisfied by energy currents. If one has an energy--momentum tensor associated to a theory, it provides a natural way to construct such energy currents~\cite{Alinhac}.

Taking a metric $g+\epsilon h$ and formally linearising the vacuum Einstein equation~\eqref{VE} yields the linearised vacuum Einstein equation, 
\begin{align}\label{LVEE}
g^{cd}\nabla_c\nabla_dh_{ab}+\nabla_a\nabla_b(\mathrm{tr}_gh)-\nabla_a(\mathrm{div}h)_b-\nabla_b(\mathrm{div}h)_a+2{{{{R}_a}^c}_b}^dh_{cd}=0,
\end{align}
which is an equation for a symmetric $2$-tensor $h$, known as the linearised metric, on a background spacetime~$(\mathcal{M},g)$. Here $\nabla$ and ${R}$ are the Levi-Civita connection and Riemann tensor, respectively, of $g$. Whilst equation~\eqref{LVEE} still needs complemented with an suitable gauge choice (typically inherited from a well-posed non-linear gauge), it is the appropriate equation of study to address the linear stability of a known metric $g$. 

Constructing conservation laws for~\eqref{LVEE} is complicated. In particular, one immediately encounters the issue that there is no energy--momentum tensor associated to a solution $h$ of the linearised vacuum Einstein equation~(\ref{LVEE}). Therefore, one cannot construct a conserved current based upon an energy--momentum tensor. This does not mean the theory should be absent of conservation laws; conservation laws should still arise in the system of gravitational perturbation due to the presence of background Killing symmetries. For example, since the Schwarzschild black hole is stationary one would expect that there is a conservation law for solutions of the linearised vacuum Einstein equation~\eqref{LVEE}.

Whilst there is not an energy-momentum tensor for the linearised vacuum Einstein equation~\eqref{LVEE}, it turns out there is a {systematic} method of deriving conservation laws for~\eqref{LVEE} originating from a work of Friedman~\cite{Friedman78} (see also~\cite{Chandrasekhar2,Chandrasekhar3}) and extended by Hollands and Wald~\cite{HolW}. As above, these conservation laws can be viewed as arising from an energy current for linearised vacuum Einstein equation~(\ref{LVEE}) which is divergence-free when associated with a Killing vector field~$X$ of the black hole in question. Since it is divergence-free when associated to symmetry, applying the divergence theorem in regions of the black hole exterior gives rise to a conservation law for an energy. Rather than being constructed from an energy-momentum tensor, the current is constructed from a symplectic form on the space of solutions to~\eqref{LVEE}.
\subsubsection{The Symplectic Current for the Wave Equation}
The symplectic form described above has a natural counterpart for the wave equation,
\begin{align}\label{eq:WE}
\Box_g\Psi=0,
\end{align}
which is an instructive example we shall now briefly discuss. The symplectic form on the space of solutions to~\eqref{eq:WE} is
\begin{align}\label{eq:sympf}
j[\Phi,\Psi]_a\doteq \Phi\nabla_a\Psi-\Psi\nabla_a\Phi.
\end{align}
If $\Phi$ and $\Psi$ solve the wave equation~\eqref{eq:WE}, this current is conserved. The current in equation~\eqref{eq:sympf} may not look immediately helpful since it involves two different solutions. However, if one has a Killing symmetry $X$ then one can set $\Phi=\mathcal{L}_X\Psi$, since $\mathcal{L}_X\Psi$ will solve the wave equation. If one does this for the Killing field $T=\partial_t$ on the region $\mathcal{R}$ in figure~\ref{fig1} on the exterior of Schwarzschild, one finds the following conservation law 
\begin{align}\label{eq:clawWE2}
\mathcal{E}^T_{u_1}[\Psi](v_0,v_1)+\mathcal{E}^T_{v_1}[\Psi](u_0,u_1)=\mathcal{E}^T_{u_0}[\Psi](v_0,v_1)+\mathcal{E}^T_{v_0}[\Psi](u_0,u_1)
\end{align}
with
\begin{align*}
\mathcal{E}^T_u[\Psi](v_0,v_1)&\doteq\frac{1}{2}\int_{v_0}^{v_1}\int_{\mathbb{S}^2}\Big(\partial_u\Psi\partial_v\Psi+|\partial_v\Psi|^2-\Psi\partial_{v}^2\Psi-\Psi\partial_u\partial_v\Psi\Big)dv\slashed{\epsilon} \\ \mathcal{E}^T_v[\Psi](u_0,u_1)&\doteq  \frac{1}{2}\int_{u_0}^{u_1}\int_{\mathbb{S}^2}\Big(\partial_u\Psi\partial_v\Psi+|\partial_u\Psi|^2-\Psi\partial_{u}^2\Psi-\Psi\partial_u\partial_v\Psi\Big)du\slashed{\epsilon},
\end{align*}
 where $\slashed{\varepsilon}$ denotes the volume form on the sphere of radius $r(u,v)$. These energy fluxes do not look immediately useful since one does not gain control of a Sobolev-type norm, i.e. its coercivity properties are obsure. However, one can use the wave equation, which on $(\mathrm{Schw}_4,g_s)$ in $(u,v,\theta,\phi)$ coordinates reduces to
\begin{align*}
-\frac{1}{\Omega^2}\partial_u\partial_v\Psi+\frac{1}{r}(\partial_v-\partial_u)\Psi+{\slashed{\Delta}}\Psi=0,
\end{align*}
to integrate by parts to show, 
 \begin{align*}
  \mathcal{E}_u^T[\Psi](v_0,v_1)&=2E^T_u[\Psi](v_0,v_1)-\int_{\mathbb{S}^2}\mathcal{A}[\Psi](u,v,\theta,\phi)\slashed{\varepsilon}\Big|_{v_0}^{v_1}\\
  \mathcal{E}_v^T[\Psi](u_0,u_1)&=2E^T_v[\Psi](u_0,u_1)+\int_{\mathbb{S}^2}\mathcal{A}[\Psi](u,v,\theta,\phi)\slashed{\varepsilon}\Big|_{u_0}^{u_1}
  \end{align*}
  where
  \begin{align*}
  E^T_u[\Psi](v_0,v_1)&\doteq \int_{v_0}^{v_1}\int_{\mathbb{S}^2}\Big(|\partial_v\Psi|^2+D|\ns\Psi|^2\Big) dv\slashed{\varepsilon},\quad E^T_v[\Psi](u_0,u_1)\doteq\int_{u_0}^{u_1}\int_{\mathbb{S}^2}\Big(|\partial_u\Psi|^2+D|\ns\Psi|^2\Big) dv \slashed{\varepsilon},
  \end{align*}
 with $\mathcal{A}\doteq \frac{1}{4}\Psi(\partial_v\Psi-\partial_u\Psi)$. 
 Therefore, due to the cancellation of the $\mathcal{A}$ terms, 
 \begin{align}\label{eq:clawWESch}
E^T_{u_1}[\Psi](v_0,v_1)+E^T_{v_1}[\Psi](u_0,u_1)=E^T_{u_0}[\Psi](v_0,v_1)+E^T_{v_0}[\Psi](u_0,u_1),
\end{align}
which is precisely the usual $T$-energy conservation law for the wave equation arising from the energy-momentum tensor. Crucially this gives rise to the coercive energy estimate
 \begin{align*}
 E^T_{u_1}[\Psi](v_0,v_1)+E^T_{v_1}[\Psi](u_0,u_1)\leq  E_{data}^T[\Psi]\doteq E^T_{u_0}[\Psi](v_0,\infty)+E^T_{v_0}[\Psi](u_0,\infty)
 \end{align*}
 for all $u_0\leq u_1< \infty$ and $v_0\leq v_1< \infty$.
\subsubsection{The Symplectic Form for Linearised Vacuum Einstein Equation}
The generalisation of the symplectic current~\eqref{eq:sympf} to the linearised vacuum Einstein equation~\eqref{LVEE} is
\begin{align}\label{eqn:SymplecticCurrentLVEE}
j[h_1,h_2]^a\doteq&P^{abcdef}\Big((h_2)_{bc}\nabla_d(h_1)_{ef}-(h_1)_{bc}\nabla_d(h_2)_{ef}\Big)
\end{align}
where
\begin{align}\label{eqn:P}
P^{abcdef}\doteq g^{ae}g^{bf}g^{cd}-\frac{1}{2}g^{ad}g^{be}g^{cf}-\frac{1}{2}g^{ab}g^{ef}g^{cd}-\frac{1}{2}g^{ae}g^{df}g^{bc}+\frac{1}{2}g^{ad}g^{ef}g^{bc}.
\end{align}
If $h_1$ and $h_2$ solve~\eqref{LVEE} then this current is conserved. As with the wave equation example above, if $X$ is a Killing symmetry and $h$ is a solution to equation~\eqref{LVEE}, then one obtains a conserved current for $h$, which we will denote $\mathcal{J}^X[h]$, by setting $h_1=h$ and $h_2=\mathcal{L}_Xh$, since $\mathcal{L}_Xh$ will also solve~\eqref{LVEE}, i.e. $\mathcal{J}^X[h]=j[h,\mathcal{L}_Xh]$. 

In~\cite{HolW}, Hollands and Wald use $\mathcal{J}^X[h]$ to construct a conservation law for linearised gravity on stationary black holes by picking $X=T\doteq\partial_t$, where $T$ is the Killing field associated to stationarity. They call the associated energy,~$\mathcal{E}^T[h]$, that arises on the boundary of the region where one applies the divergence theorem, the `canonical energy' for the linearised vacuum Einstein equation~(\ref{LVEE}). In~\cite{HolW}, the authors introduce a linear stability criterion associated to~$\mathcal{E}^T[h]$ evaluated on a Cauchy hypersurfaces. In rough terms, for a class of admissible initial data, the positivity or negativity of~$\mathcal{E}^T[h]$ (evaluated on a Cauchy hypersurface) is an indication of whether a stationary vacuum black hole spacetime linearly stable or unstable, respectively.

The criterion is appealing conceptually since one can write down a conservation law for \textit{any} stationary vacuum black hole spacetime (or more generally a black hole with a symmetry) and one `simply' needs to check the sign of~$\mathcal{E}^T[h]$ for a particular class of admissible data. However, in practice, it is hard to establish the sign of the canonical energy. This is in no small part due to the fact one has to make sure that the initial data one is considering satisfies the aforementioned admissibility criteria. Indeed, even for the `basic' case of the stability of the $4$-dimensional Schwarzschild black hole, the original work of Hollands and Wald did not establish positivity (the Schwarzschild black hole is after all (non-)linearly stable~\cite{DHR,DHRT}). This was partially rectified in the work of Prabu and Wald~\cite{PW2} where they show that the canonical energy of a metric perturbation arising from a `Hertz potential' is positive. Crucially, if one is interested in proving energy boundedness statements, the coercivity properties are even more obscure. 

\subsubsection{Conservation Laws in Double Null Gauge and Revealing Coercivity}

In~\cite{Holzegel}, Holzegel constructed two conservation laws for linearised gravity in \textit{double null gauge} on a characteristic rectangle in the Schwarzschild exterior as depicted in the figure~\ref{fig1}. One of these conservation laws occurs at the level of linearised Ricci coefficients and the other involving linearised curvature. Remarkably, Holzegel spots them by eye in the double null decomposed linearised vacuum Einstein equations. Moveover, in~\cite{Holzegel}, coercivity properties of these conservation laws are established to attain an energy boundedness statement for the linearised shear on any outgoing null cone, $C_u$, and the linearised shear on null infinity. This has now been extended in~\cite{HolzegelSC} to obtain energy boundedness and pointwise boundedness for the gauge invariant Teukolsky quantities on Schwarzschild, which is a key ingredient in approaching the stability problem for Schwarzschild (and, also, Kerr)~\cite{DHR,DHR2,ABBSM19,DHRT,Szeftel,YakovRita2,Elena}. 

The main question that arises from the above discussion:
\begin{Question*}
Is there a connection between Holzegel's conservation laws and the canonical energy?
\end{Question*}
This paper answers this question in the affirmative. In particular, the main novelties of this paper are
\begin{enumerate}
\item a systematic derivation of the conservation laws in~\cite{Holzegel} for double null decomposed gravitational perturbations \textit{from the canonical energy}. The advantage of this is that the canonical energy can be written down on any stationary spacetime, i.e. one starts with a conservation law which one can manipulate into a desirable form, rather than having to construct a conservation law from scratch. This provides one with a natural path to investigate energy boundedness on other spacetimes, for example Kerr or Schwarzschild--Tangherlini.\footnote{See~\cite{MyThesis} for a generalization to higher dimensional Schwarzschild--Tangherlini.} Further, this provides a resolution to the issue of coercivity and positivity of canonical energy~$\mathcal{E}^T[h]$ on the Schwarzschild exterior; the precise correspondence between canonical energy and Holzegel's conservation laws provided in this paper allows one to interpret the stability results established in~\cite{HolzegelSC} as results about the canonical energy.
\item a construction of a novel conservation law for the decoupled Teukolsky null curvature components $\alpha$ and $\underline{\alpha}$.
\end{enumerate}

In short, the first point above are addressed by decomposing the canonical energy in double null gauge. The conservation law for $\mathcal{E}^T[h]$ is then understood locally on the characteristic rectangle in figure~\ref{fig1} above. It is shown that the canonical energy conservation law is equivalent to Holzegel's conservation law for the energy involving linearised connection coefficients. 

The remainder of this introduction is organised as follows. Section~\ref{OverviewIntro} gives an overview of the main results in this paper and section~\ref{Outlook} discusses some possible directions for extension to this work.

\subsection{Overview of Main Results: Theorems~\ref{thm:CanEntoGusCon}-\ref{TH5}}\label{OverviewIntro}
Section~\ref{CanEn} gives a more in depth exposition on the canonical energy. The discussion here differs from the work of Hollands and Wald~\cite{HolW} where the canonical energy is viewed as a constrained variational principle evaluated on Cauchy hypersurfaces. In this work, the $X$-canonical energy,  $\mathcal{E}^X_{\Sigma}[h]$, on some hypersurface $\Sigma$ is simply viewed as the energy quantity for the linearised metric $h$ arising from the vector field current $\mathcal{J}[h]^X$ (see equation~\eqref{eqn:SymplecticCurrentLVEE}) which can be evaluated \textit{locally} and is divergence free if $X$ is Killing and $h$ satisfies the linearised vacuum Einstein equation~\eqref{LVEE}. Now, since $\mathcal{J}[h]^X$ is divergence free if $X$ is Killing, the divergence theorem allows one to construct conservation laws for $\mathcal{E}^X[h]$ associated to the boundary of some spacetime region. Additionally, a modified definition of energy current for~\eqref{LVEE} is proposed in section~\ref{sec:NCurr} in close analogy with the currents that arise from the energy-momentum tensor for the wave equation~\eqref{eq:WE}.

Section~\ref{DNG} contains a brief introduction and review of the Einstein equation in double null gauge as is neccesary for the computations in section~\ref{Results}. In short, a solution to linearised vacuum Einstein equation~\eqref{LVEE} in double null gauge is encoded by geometric quantities on surfaces, $\mathcal{S}_{u,v}$, homeomorphic to spheres. One has the linearised metric quantities
\begin{align*} 
\mathrm{tr}\slashed{h}, \slashed{h}_{AB}, \bmlin_A,\Olin,
\end{align*}
the linearised connection coefficients
\begin{align*} 
\otx,\otxb,\olin,\olinb,\elin_A,\eblin_A, \xlin_{AB},\xblin_{AB} ,
\end{align*}
and the linearised curvature components
\begin{align*}
\rlin,\slin,\blin_A,\bblin_A, \alin_{AB},\ablin_{AB},
\end{align*}
where the roman indices denote the  $\mathcal{S}_{u,v}$-tensor type. The linearised vacuum Einstein equation~\eqref{LVEE} is then encoded by linearising the Levi-Civita connection condition (this gives a set of equations known as the linearised null structure equations) and linearised Bianchi identities. 

Section~\ref{Results} contains the main body of original work. By evaluating the canonical energy associated to the Killing field $T\doteq\partial_t$ for the Schwarzschild spacetime on the characteristic rectangle depicted in figure~\ref{fig1} and choosing the metric perturbation $h$ to be in double null gauge (see definition~\ref{DNGDefinition}), a double null decomposition of the canonical energy is achieved. Since $T$ is a Killing vector field one obtains a conservation law for the $T$-canonical energy
\begin{align}\label{CEconlaw}
\mathcal{E}_{u_1}^T[h](v_0,v_1)+\mathcal{E}_{v_1}^T[h](u_0,u_1)=\mathcal{E}_{u_0}^T[h](v_0,v_1)+\mathcal{E}_{v_0}^T[h](u_0,u_1),
\end{align}
where, to ease notation, one denotes
\begin{align*}
    \mathcal{E}_{u}^T[h](v_0,v_1)\doteq \mathcal{E}_{C_{u}\cap\{ v_0\leq v\leq v_1\}}^T[h],\quad
    \mathcal{E}_{v}^T[h](u_0,u_1)\doteq \mathcal{E}_{C_{v}\cap \{ u_0\leq u\leq u_1\}}^T[h].
\end{align*} 
The main theorem on the double null decomposition of the canonical energy is the following:
\begin{theorem}\label{thm:CanEntoGusCon}
Suppose $h$ is a smooth solution to the linearised vacuum Einstein equation~(\ref{LVEE}) in double null gauge on the Schwarzschild black hole exterior. Then the $T$-canonical energy of $h$ on the null cones $C_u\cap \{ v_0\leq v\leq v_1\}$ and $\underline{C}_v\cap \{ u_0\leq u\leq u_1\}$ is given by
\begin{align}
{\mathcal{E}}^T_u[h](v_0,v_1)&=2F_u[\Gamma](u_0,u_1)-2\int_{\mathbb{S}_{u,v}^2}\mathcal{A}[h](u,v,\theta,\phi)\slashed{\varepsilon}\Big|_{v_0}^{v_1},\label{RelationGW1}\\
{\mathcal{E}}^T_v[h](u_0,u_1)&=2F_v[\Gamma](v_0,v_1)+2\int_{\mathbb{S}_{u,v}^2}\mathcal{A}[h](u,v,\theta,\phi)\slashed{\varepsilon}\Big|_{u_0}^{u_1}.\label{RelationGW2}
\end{align}
where
\begin{align}
F_u[\Gamma](v_0,v_1)&\doteq\int_{v_0}^{v_1}\int_{\mathbb{S}_{u,v}^2}\Big[|\Omega\xlin |^2+2|\Omega\eblin |^2-2\olin \otxb -\frac{1}{2}{\otx^2}+4\omega\Big(\frac{\Olino }{\Omega}\Big)\otx\Big]dv\slashed{\varepsilon},\label{F1}\\
F_v[\Gamma](u_0,u_1)&\doteq\int_{u_0}^{u_1}\int_{\mathbb{S}_{u,v}^2}\Big[|\Omega\xblin |^2+2|\Omega\elin |^2-2\olinb \otx -
\frac{1}{2}\otxb ^2-4\omega\Big(\frac{\Olino }{\Omega}\Big)\otxb \Big] du\slashed{\varepsilon}.\label{F2}
\end{align}
and with 
\begin{align}\label{BasicA}
\mathcal{A}[h]\doteq &\frac{1}{4}(\olinb -\olin )\mathrm{tr}\slashed{h}-\frac{1}{4}(\elin -\eblin )(\bmlin)+\frac{1}{8}\Big[\otxb -\otx\Big]\mathrm{tr}\slashed{h}-\frac{\Omega}{4}\langle\xblin -\xlin ,\hat{\slashed{h}}\rangle\\
&\nonumber+\frac{3}{2}\Big(\frac{\Olino }{\Omega}\Big)\Big[\otx -\otxb \Big]+\frac{\Omega\mathrm{tr}\chi}{2}\Big(\frac{\Olino }{\Omega}\Big)\Big(\mathrm{tr}\slashed{h}-4\Big(\frac{\Olino }{\Omega}\Big)\Big)\nonumber.
\end{align}
Moreover, the following conservation law is satisfied:
\begin{align}\label{MODCE2}
F_{u_1}[\Gamma](v_0,v_1)+F_{v_1}[\Gamma](u_0,u_1)=F_{u_0}[\Gamma](v_0,v_1)+F_{v_0}[\Gamma](u_0,u_1).
\end{align}
\end{theorem}
The equation~(\ref{MODCE2}) in this theorem follows the from the equations~(\ref{RelationGW1}) and~(\ref{RelationGW2}) in conjunction with the conservation law for the canonical energy~(\ref{CEconlaw}). This conservation law in~\eqref{MODCE2} is precisely Holzegel's conservation law from~\cite{Holzegel}. Therefore, one can view this theorem as a derivation of Holzegel’s conservation law from the canonical energy.

The proof of theorem~\ref{thm:CanEntoGusCon} (which can be found in section~\ref{TMR}) relies upon using the linearised null structure equations in propositions~\ref{LinMet}-\ref{LinCod} (but not the linearised Bianchi equations in proposition~\ref{LinBianchi}) to integrate by parts and produce the boundary term $\mathcal{A}$. Moreover, in proving this result the \textit{linearised Gauss and Codazzi constraint equations are key} (see proposition~\ref{LinGauss} and see proposition~\ref{LinCod}). This directly relates Hollands and Wald's admissible data criterion: the linearised contraint equations must be verified. However, one advantage of using the double null decomposition is that one can readily use these linearised constraint equations.

Section~\ref{betaconsproof} constructs a `higher order' $T$-canonical energy conservation law associated to $\mathcal{L}_{\Omega_i}h$ where $\{\Omega_i\}_{i=1}^3$ are the Killing vector fields associated to the $SO(3)$ symmetry of the Schwarzschild spacetime (see equation~\eqref{SO3} for explicit expressions for~$\{\Omega_i\}_{i=1}^3$). For convenience, define the following fluxes:
\begin{align}
F_u[\mathcal{R}](v_0,v_1)\doteq\int_{v_0}^{v_1}\int_{\mathbb{S}_{u,v}^2}\Big(&\frac{\Omega^2r^2}{2}|\blin |^2+\frac{3\Omega^2 r^2\rho}{2}|\eblin |^2+\frac{\Omega^2r^2}{2}\big(|\slin |^2+|\rlin |^2\big)-\frac{3r^2\rho}{2}\olin \otxb \label{FA1}\\
&\nonumber-\frac{3r^2\rho}{2}\Big[\frac{1}{2}(\Omega\mathrm{tr}\chi)-2\omega \Big]\Big(\frac{\Olino }{\Omega}\Big)\accentset{(1)}{(\Omega\mathrm{tr}{\chi})}\Big)dv\slashed{\varepsilon},
\end{align}
\begin{align}
F_v[\mathcal{R}](u_0,u_1)\doteq\int_{u_0}^{u_1}\int_{\mathbb{S}_{u,v}^2}\Big(&\frac{\Omega^2r^2}{2}|\bblin |^2+\frac{3\Omega^2 r^2\rho}{2}|\elin |^2+\frac{\Omega^2r^2}{2}\big(|\slin |^2+|\rlin |^2\big)-\frac{3r^2\rho}{2}\olinb \otx\label{FA2}\\
&\nonumber+\frac{3r^2\rho}{2}\Big[\frac{1}{2}(\Omega\mathrm{Tr}_{\slashed{g}}\chi)-2\omega \Big]\Big(\frac{\Olino }{\Omega}\Big)\otxb \Big)du\slashed{\varepsilon},
\end{align}

With these definitions in hand, the following theorem is proved in section~\ref{Results}.
\begin{theorem}\label{thm:CEntoGConb}
Suppose $h$ is a smooth solution of the linearised vacuum Einstein equation~(\ref{LVEE}) in double null gauge on the Schwarzschild black hole exterior. The $T$-canonical energy of $\mathcal{L}_{\Omega_k}h$ satisfies
\begin{align*}
\sum_k{\mathcal{E}}^T_u[\mathcal{L}_{\Omega_k}h](v_0,v_1)&=8F_u[\mathcal{R}](v_0,v_1)+4F_u[\Gamma](v_0,v_1)-2\int_{\mathbb{S}_{u,v}^2}\mathcal{B}(u,v,\theta,\phi)\slashed{\varepsilon}\Big|_{v_0}^{v_1},\\
\sum_k{\mathcal{E}}^T_v[\mathcal{L}_{\Omega_k}h](u_0,u_1)&=8F_v[\mathcal{R}](u_0,u_1)+4F_v[\Gamma](u_0,u_1)+2\int_{\mathbb{S}_{u,v}^2}\mathcal{B}(u,v,\theta,\phi)\slashed{\varepsilon}\Big|_{u_0}^{u_1}.
\end{align*}
with
\begin{align*}
\mathcal{B}[h]&\doteq {r^2}\Big(\otx \divs \elin -\otxb \divs \eblin -\rlin \big(\otx -\accentset{(1)}{(\Omega\mathrm{tr}{\underline{\chi}})}\big)+2\Omega\mathrm{tr}\chi\Big(\frac{\Olino }{\Omega}\Big)\rlin \\
&\nonumber+(\Omega\mathrm{tr}\chi)\langle\elin ,\eblin \rangle+\Big[\frac{\omega}{\Omega^2}-\frac{\mathrm{tr}\chi}{2\Omega}\Big]\otxb \otx \Big)+\sum_k\mathcal{A}[\mathcal{L}_{{\Omega}_k}h],
\end{align*}
where $\mathcal{A}[\mathcal{L}_{{\Omega}_k}h]$ results from expressing $\mathcal{A}[h]$ in equation~(\ref{BasicA}) in terms of $h$ and replacing it with $\mathcal{L}_{{\Omega}_k}h$. Moreover, the following conservation law is satisfied:
\begin{align}\label{Holzegel2}
F_{u_1}[\mathcal{R}](v_0,v_1)+F_{v_1}[\mathcal{R}](u_0,u_1)=F_{u_0}^T[\mathcal{R}](v_0,v_1)+F_{v_0}^T[\mathcal{R}](u_0,u_1).
\end{align}
\end{theorem}
The reader should note that Holzegel also identified the conservation law~(\ref{Holzegel2}) in~\cite{Holzegel} by eye (see equations (82) and (83) and proposition 8.1 in~\cite{Holzegel}).

Section~\ref{alphacons} constructs an additional `higher order' $T$-canonical energy conservation law associated to~$\mathcal{L}_Th$. For convenience define the following fluxes:
\begin{align*}
\dot{F}_u[\mathcal{R}](v_0,v_1)&\doteq\int_{v_0}^{v_1}\int_{\mathbb{S}_{u,v}^2}\Big(\frac{\Omega^4}{4}|\alin |^2+\frac{3}{2}\Omega^4\big(|\rlin |^2+|\slin |^2+|\blin |^2\big)+\frac{\Omega^4}{2}|\bblin |^2+f_2|\xblin |^2+f_1|\xlin |^2+f_3|\eblin |^2
\\
&\nonumber\qquad\qquad-\frac{f_3}{\Omega^2}\olin \otxb +\frac{2}{\Omega^2}\Big(\omega f_3+{2\Omega\mathrm{tr}\chi f_2}\Big)\Big(\frac{\Olino }{\Omega}\Big)\otx -\frac{f_1}{2\Omega^2}\otx ^2\\
&\nonumber\qquad\qquad-\frac{f_2}{2\Omega^2}\Big[\otxb +2(\Omega\mathrm{tr}\chi)\Big(\frac{\Olino }{\Omega}\Big)\Big]^2\Big)dv\slashed{\varepsilon},\\
\dot{F}_v[\mathcal{R}](u_0,u_1)&\doteq\int_{u_0}^{u_1}\int_{\mathbb{S}_{u,v}^2}\Big(\frac{\Omega^4}{4}|\ablin |^2+\frac{3}{2}\Omega^4\big(|\rlin |^2+|\slin |^2+|\bblin |^2\big)+\frac{\Omega^4}{2}|\blin |^2+f_1|\xblin |^2+f_2|\xlin |^2+ f_3|\elin |^2\\
&\nonumber\qquad\qquad-\frac{f_3}{\Omega^2}\olinb \otx  -\frac{2}{\Omega^2}\Big(\omega f_3+{2\Omega\mathrm{tr}\chi f_2}\Big)\Big(\frac{\Olino }{\Omega}\Big)\otxb -\frac{f_1}{2\Omega^2}\otxb ^2\\
&\nonumber\qquad\qquad-\frac{f_2}{2\Omega^2}\Big[\otx -2(\Omega\mathrm{tr}\chi)\Big(\frac{\Olino }{\Omega}\Big)\Big]^2\Big)du\slashed{\varepsilon}
\end{align*}
with
\begin{align*}
f_1&\doteq -\Omega^2\Big(\omega^2+\frac{5}{4}\Omega^2\rho\Big),\quad
f_2\doteq -\frac{3}{4}\Omega^4\rho,\quad
f_3\doteq 2\Omega^2(\Omega^2\rho-\omega^2).
\end{align*}
Additionally, the following theorem is proved in section~\ref{alphacons}:
\begin{theorem}\label{TH5}
Suppose $h$ is a smooth solution of the linearised vacuum Einstein equation~(\ref{LVEE}) in double null gauge on the Schwarzschild black hole exterior. Then the following conservation law is satisfied:
\begin{align}\label{alphaconservationlaw}
\dot{F}_{u_1}[\mathcal{R}](v_0,v_1)+\dot{F}_{v_1}[\mathcal{R}](u_0,u_1)=\dot{F}_{u_0}[\mathcal{R}](v_0,v_1)+\dot{F}_{v_0}[\mathcal{R}](u_0,u_1).
\end{align}
\end{theorem}
\begin{remark}
To the best of the authors knowledge, the conservation law~\eqref{alphaconservationlaw} has not been previously derived. 
\end{remark}
\begin{remark}
It is reasonable to expect that one could apply this conservation law in the spirit of~\cite{HolzegelSC} to establish coercivity and, therefore, gain $L^2$-control of $\alin$ directly. This relies upon the precise transformation of $\dot{F}_{u}[\mathcal{R}](v_0,v_f)$ under a change of gauge, which one would have to work out. 
\end{remark}

\subsection{Outlook}\label{Outlook}
It is conjectured that the subextremal Kerr black hole spacetime is asymptotically stable as a solution to the vacuum Einstein equation (see section IV of the introduction of~\cite{DHRT} for a precise formulation of this conjecture and the works~\cite{KS3,GKS} on the non-linear stability in slowly rotating regime). In view of the works~\cite{DHR2} and~\cite{YakovRita1,YakovRita2} on quantitative boundedness and decay for the Teukolsky equation on subextremal Kerr, the full linear stability of the subextremal Kerr spacetime is within reach in analogy with~\cite{DHR}. However, this work in conjuction with~\cite{Holzegel,HolzegelSC} provides an alternative path to produce energy boundedness for the subextremal Kerr spacetime. The expectation in the Kerr case is that one could prove a restricted stability result for \textit{axisymmetric perturbations}. One should also note the work of Moncrief and Gudapati~\cite{Moncrief2} in this direction.

Another interesting open direction is the extension of these results to other matter models. The canonical energy arises naturally from the Einstein--Hilbert action for the Einstein vacuum equation by considering antisymmetrised variations of the action. The notion of canonical energy extends naturally to many theories with a Lagrangian formulation (see Keir~\cite{Joe}). Therefore, the canonical energy as formulated in this aper could be used to establish stability statements outside of vacuum. 

\section*{Acknowledgements}
First and foremost, the author would like to thank Gustav Holzegel for many insightful discussions during this project and for useful comments on the manuscript. Thanks also go to Mihalis Dafermos for the impetus for this project and helpful comments. Additionally, the auther is grateful to Martin Taylor and Claude Warnick who gave constructive feedback on an early version of this work that appeared in the author's Ph.D. thesis.

The author acknowledges that the majority of the research presented in this work was carried out as a Ph.D. student in the Cambridge Centre for Analysis at the University of Cambridge. In particular, support for this work came from UK Engineering and Physical Sciences Research Council (EPSRC) Grant No. EP/L016516/1 for the University of Cambridge Centre for Doctoral Training, the Cambridge Centre for Analysis.

\section{Canonical Energy and Conserved Currents}\label{CanEn}

In this section a pedestrian approach to the canonical energy is presented. The canonical energy for the linearised vacuum Einstein equation~(\ref{LVEE}) is defined in section~\ref{CanLVEE}. In section~\ref{sec:HOCE}, higher order canonical energies are discussed. Finally, in section~\ref{sec:NCurr}, a novel current is defined for the linearised vacuum Einstein equation~\eqref{LVEE}, which will see application in future work.

Note that the exposition in this section gives a different viewpoint on the canonical energy from that of the original work of Hollands and Wald. In particular, the canonical energy is simply viewed as a quantity for the linearised metric $h$ arising from a `vector field current' which can be evaluated \textit{locally}. This section does not discuss the nice connection of the canonical energy to black hole thermodynamics or the stability criterion associated to the canonical energy when evaluated on Cauchy hypersurfaces. The interested reader should consult~\cite{HolW}.

\subsection{Canonical Energy for the Linearised Vacuum Einstein Equation}\label{CanLVEE}
It turns out to be convenient to rewrite the linearised vacuum Einstein equation~(\ref{LVEE}) (plus its trace) in the following form:
\begin{prop}\label{PrincipleLVEE}
Suppose $(\mathcal{M},g)$ satisfies the vacuum Einstein equation~(\ref{VE}) and $h$ satisfies the linearised vacuum Einstein equation~(\ref{LVEE}). Then
\begin{align}\label{eqn:LVEEP}
{{P^{a}}_{(bc)}}^{def}\nabla_a\nabla_dh_{ef}=0,
\end{align}
where $P$ is defined in equation~\eqref{eqn:P}.
\end{prop}
\begin{proof}
Computing directly, using equation~(\ref{eqn:P}), gives
\begin{align*}
{{P^{a}}_{(bc)}}^{def}\nabla_a\nabla_dh_{ef}&=\nabla^a\nabla_{(c}h_{b)a}-\frac{1}{2}\Box_gh_{bc}-\frac{1}{2}\nabla_b\nabla_c\mathrm{Tr}_{g}h-\frac{1}{2}\mathrm{div}(\mathrm{div}h)g_{bc}+\frac{1}{2}(\Box_g\mathrm{tr}_gh)g_{bc}.
\end{align*}
The last two terms are the trace of the linearised vacuum Einstein equation~(\ref{LVEE}) and, therefore, cancel. Moreover, one can apply the Ricci identity to the the first term to give
\begin{align*}
{{P^{a}}_{(bc)}}^{def}\nabla_a\nabla_dh_{ef}=\nabla_{(c}\mathrm{div}h_{b)}-{{{R_{(b}}^a}_{c)}}^dh_{ad}-\frac{1}{2}\Box_gh_{bc}-\frac{1}{2}\nabla_b\nabla_c\mathrm{tr}_{g}h,
\end{align*}
where one uses $\mathrm{Ric}[g]=0$. Therefore, using the symmetries of the Riemann tensor and the linearised vacuum Einstein equation~(\ref{LVEE}) one has the result. 
\end{proof}

\begin{prop}\label{closedexp}
Suppose $(\mathcal{M},g)$ satisfies the vacuum Einstein equation and $h_1$ and $h_2$ satisfy the linearised vacuum Einstein equation~(\ref{LVEE}). Then the symplectic current $j[h_1,h_2]$ defined in equation~\eqref{eqn:SymplecticCurrentLVEE} is divergence free.
\end{prop}
\begin{proof}
Follows from a direct computation using proposition~\ref{PrincipleLVEE} and noting that $P$ has the symmetry
\begin{align*}
P^{abcdef}=P^{dfeacb}.
\end{align*} 
\end{proof}
\begin{remark}
As Hollands and Wald illustrate in~\cite{HolW}, the current ${j}[h_1,h_2]$ arises naturally by considering antisymmeterised second variations of the Einstein--Hilbert action for a vacuum spacetime $(\mathcal{M},g)$ which has Lagrangian density
\begin{align*}
L[g]\doteq\frac{1}{16\pi}\mathrm{Scal}[g]\varepsilon,
\end{align*}
where $\mathrm{Scal}[g]$ is the Ricci scalar of $g$ and $\varepsilon$ is the volume form associated to $g$. 
\end{remark}
\begin{definition}[Canonical Energy]\label{CanonicalEnergyDefinition}
Let $(\mathcal{M},g)$ be a spacetime satisfying the vacuum Einstein equation~(\ref{VE}) and $h$ be a solution to the linearised vacuum Einstein equation~(\ref{LVEE}). Let $X$ be a vector field and $\Sigma$ be a hypersurface in $(\mathcal{M},g)$ with (unit) normal $n$. The $X$-canonical energy current is defined as
\begin{align}\label{XCEC}
\mathcal{J}^X[h]\doteq j[h,\mathcal{L}_Xh]
\end{align}
with the associated $X$-canonical energy on $\Sigma$ is defined as
\begin{align*}
\mathcal{E}^X_{\Sigma}[h]\doteq \int_{\Sigma}n(\mathcal{J}^X[h])\mathrm{dvol}_{\Sigma}.
\end{align*}
\end{definition}
\begin{remark}
This definition extends the definition of the canonical energy given by Hollands and Wald in the sense that~\cite{HolW} considers the case only where $(\mathcal{M},g)$ is stationary with stationary Killing field $T$ and then take~$X=T$.
\end{remark}

If $h$ is a smooth solution to the linearised vacuum Einstein equation~(\ref{LVEE}) and $X$ a Killing vector field applying the divergence theorem in conjunction with proposition~\ref{closedexp} to $\mathcal{J}^X[h]$ between two homologous hypersurfaces $\Sigma_1$ and $\Sigma_2$ in a vacuum spacetime yields a conservation law
\begin{align*}
\mathcal{E}^X_{\Sigma_2}[h]=\mathcal{E}^X_{\Sigma_1}[h].
\end{align*}
\begin{remark}
If $X$ is \textit{not} a Killing vector field then the above can be modified to include $\mathrm{div}(\mathcal{J}^X[h])\neq 0$ as a spacetime bulk term. This could be useful if one wanted to produce a Morawetz-type~\cite{Morawetz} spacetime estimate. 
\end{remark}
 
\subsection{Higher Order Canonical Energies}\label{sec:HOCE}
Suppose the stationary spacetime $(\mathcal{M},g)$ has a Killing field $k$ (not necessarily the stationary Killing vector field $T$). If $h$ is a solution to the linearised vacuum Einstein equation~(\ref{LVEE}) then $\mathcal{L}_kh$ is also a solution i.e. one can commute as many Lie derivatives $\mathcal{L}_k$ through the linearised vacuum Einstein operator as one likes. Hence, one can consider a `higher order' $X$-canonical energy $\mathcal{E}^X_{\Sigma}[\mathcal{L}^m_kh]$ resulting from
\begin{align*}
(\mathcal{J}^X[\mathcal{L}^m_kh])^a=P^{abcdef}\Big((\mathcal{L}_X\mathcal{L}^m_kh)_{bc}\nabla_d(\mathcal{L}_k^mh)_{ef}-(\mathcal{L}_k^mh)_{bc}\nabla_d(\mathcal{L}_X\mathcal{L}_k^mh)_{ef}\Big).
\end{align*}
For the $4$-dimensional Schwarzschild solution one has the Killing fields associated to spherical symmetry $\Omega_i$ $i=1,2,3$. In sections~\ref{betaconsproof} and~\ref{alphacons}, the $T$-canonical energies of $\mathcal{L}_{\Omega_i}h$ and $\mathcal{L}_{T}h$ will be evaluated.
\subsection{A Modified Canonical Energy Current}\label{sec:NCurr}
This section is slightly peripheral to the main topic of this paper. Namely, this section constructs a new current for the linearised vacuum Einstein equation in close analogy with the currents that arise from the energy--momentum tensor for the wave equation~\cite{Alinhac}. The claim is that the current defined in proposition~\ref{MorConserved} is computationally advantageous over the canonical energy current. See~\cite{MyThesis} for an application to Schwarzschild--Tangherlini.

A way around the problem of a lack of energy--momentum tensor is simply to abandon the view point that these currents arise from a energy--momentum tensor and approach with a `vector field multiplier' method. This method proceeds as follows. Let $X$ be a vector field and suppose $f$ is a scalar or a tensor on a spacetime $(\mathcal{M},g)$ which solves some linear equation
\begin{align}
    \mathfrak{D}_{g}f=0,\label{ExampleEquation}
\end{align}
where $\mathfrak{D}_g$ is some differential operator depending on the metric. Then one can try to construct an `$X$-energy' for the equation~(\ref{ExampleEquation}) by multiplying the equation by $\mathcal{L}_X(f)$ and trying to write the expression as a total divergence plus terms that vanish if $X$ is a Killing symmetry of the spacetime. For example, one can consider constructing an energy for solution $\Phi\in C^{\infty}(M)$ to the wave equation~\eqref{eq:WE} in this manner. Multiplying the equation by $\mathcal{L}_X(\Phi)=X(\Phi)$ and collecting derivatives gives
\begin{align}
    0&=X(\Phi)\Box_g\Phi=\mathrm{div}(J^X[\Phi])-({\Pi}^X)^{ab}\Big(\nabla_a\Phi\nabla_b\Phi-\frac{1}{2}g_{ab}|\nabla\Phi|^2_g\Big),\nonumber
\end{align}
where
\begin{align}\label{eqn:WECurrent}
J^X[\Phi]_a&\doteq X(\Phi)\nabla_a\Phi-\frac{1}{2}X_a|\nabla\Phi|^2_g,
\end{align}
and $\Pi^X=\frac{1}{2}\mathcal{L}_Xg$ is the deformation tensor. Here, $J^X[\Phi]_a$ is the usual current arising from the energy--momentum tensor. Therefore, if one had no knowledge of the energy--momentum tensor, one could produce the conservation law for arising from $J^X[\Phi]_a$ by this simple vector field multiplier method. 

It turns out that for the linearised vacuum Einstein equation~(\ref{LVEE}) one can perform an analogous computation to this vector field multiplier view point by expressing the linearised vacuum Einstein equation (plus its trace) in the form~\eqref{eqn:LVEEP}. If one now contracts with $(\mathcal{L}_Xh)^{bc}$ and tries to write the expression as a total divergence plus terms that vanish if $X$ is a Killing symmetry, one finds the following result. 
\begin{prop}\label{MorConserved}
Let $(\mathcal{M},g)$ solve the vacuum Einstein equation~(\ref{VE}) and $h$ be a solution to the linearised vacuum Einstein equation~(\ref{LVEE}) and $X$ a Killing vector field. Let 
\begin{align}
(\mathfrak{J}^X[h])^a\doteq P^{abcdef}(\mathcal{L}_Xh)_{bc}\nabla_dh_{ef}-\frac{1}{2}X^aP(\nabla h,\nabla h),\qquad P(\nabla h,\nabla h)&\doteq P^{abcdef}\nabla_ah_{bc}\nabla_dh_{ef},\label{MorCurrent}
\end{align}
where $P$ is defined in equation~\eqref{eqn:P}. Then if $X$ is Killing then $\mathfrak{J}^X[h]$ is divergence-free. 
\end{prop}
\begin{proof}
One can show, using the Ricci identity, that, if $X$ is Killing then,
\begin{align}\label{eq:Liecom}
\nabla_a(\mathcal{L}_Xh)_{bc}=\mathcal{L}_X(\nabla h)_{abc}.
\end{align}
Using equation~\eqref{eq:Liecom}, that $\nabla_aX_b=\nabla_{[a}X_{b]}$ for a Killing vector field and the definition of the Lie derivative gives
\begin{align*}
\nabla_a(\mathcal{L}_Xh)_{bc}&=\nabla_X(\nabla_a h)_{bc}+\nabla_a{h_{b}}^{p}\nabla_{[c}X_{p]}+\nabla^ph_{bc}\nabla_{[a}X_{p]}+\nabla_a{h_{c}}^p\nabla_{[b}X_{p]}.
\end{align*}
Contracting the last three terms with $P$, one can calculate that
\begin{align*}
P^{abcdef}\nabla_a{h_{b}}^{p}\nabla_{[c}X_{p]}\nabla_dh_{ef}&=\Big(\nabla^ah^{bc}\nabla_c{h^p}_b-\frac{1}{2}(\mathrm{div}h)^p\nabla^a\mathrm{Tr}h\Big)\nabla_{[a}X_{p]},\\
P^{abcdef}\nabla^ph_{bc}\nabla_{[a}X_{p]}\nabla_dh_{ef}&=\Big(\nabla^ph^{bc}\nabla_c{h^a}_b-\frac{1}{2}\nabla^ph^{ad}\nabla_d\mathrm{Tr}h-\frac{1}{2}(\mathrm{div}h)^a\nabla^p\mathrm{Tr}h\Big)\nabla_{[a}X_{p]}\\
P^{abcdef}\nabla_a{h_{c}}^p\nabla_{[b}X_{p]}\nabla_dh_{ef}&=-\frac{1}{2}\nabla^ah^{pd}\nabla_d\mathrm{Tr}h\nabla_{[a}X_{p]}.
\end{align*}
Therefore, by symmetry, the sum of these terms vanishes. Hence, denoting $Z^a=P^{abcdef}(\mathcal{L}_Xh)_{bc}\nabla_dh_{ef}$ one has
\begin{align*}
\mathrm{div}(Z)&=P^{abcdef}\nabla_X(\nabla_a h)_{bc}\nabla_dh_{ef}.
\end{align*}
Now, using that $P^{abcdef}\nabla_a h_{bc}\nabla_dh_{ef}=P^{defabc}\nabla_a h_{bc}\nabla_dh_{ef}$ one has
\begin{align}
P^{abcdef}\mathcal{L}_X(\nabla h)_{abc}\nabla_dh_{ef}&=\mathrm{div}\Big(\frac{X}{2}P(\nabla h,\nabla h)\Big)-\frac{1}{2}(\mathrm{div}X)P(\nabla h,\nabla h)=\mathrm{div}\Big(\frac{X}{2}P(\nabla h,\nabla h)\Big),\nonumber
\end{align}
since $X$ is Killing.
\end{proof}
\begin{remark}
Note the similarity between the familar current for the wave equation~\eqref{eqn:WECurrent} and the current in equation~\eqref{MorCurrent} for the linearised vacuum Einstein equation~\eqref{LVEE}.
\end{remark}

\subsubsection{Relation to the Canonical Energy Current}\label{CEMultiRelation}
It seems reasonable to expect that the $X$-canonical energy for the linearised vacuum Einstein equation~(\ref{LVEE}) on a spacetime is related to the $X$-energy associated to the current~$\mathfrak{J}^X[h]$ constructed here. The following proposition confirms this expectation.
\begin{prop}\label{MortoCE}
Suppose $X$ is a Killing field for a vacuum spacetime $(\mathcal{M},g)$ and $h$ solves the linearised vacuum equation~(\ref{LVEE}). Then the $X$-canonical energy current~\eqref{XCEC} can be expressed as
\begin{align*}
(\mathcal{J}^X[h])^a=2(\mathfrak{J}^X[h])^a+(\mathfrak{j}^X[h])^a,
\end{align*}
where $(\mathfrak{J}^X[h])_a$ is defined in equation~(\ref{MorCurrent}) and
\begin{align*}
(\mathfrak{j}^X[h])_a\doteq (\nabla^hA)_{ah},\qquad  A^{ah}\doteq X^{[a}P^{h]bcdef}h_{bc}\nabla_{d}h_{ef},
\end{align*}
i.e., $\mathcal{J}^X[h]$ and $\mathfrak{J}^X[h]$ are related by a divergence. Moreover, $(\mathfrak{j}^X[h])_a$ is divergence free.
\end{prop}
\begin{proof}
Using that $X$ is Killing, we have
\begin{align*}
(\mathcal{J}^X[h])^a=P^{abcdef}\Big[(\mathcal{L}_Xh)_{bc}(\nabla_dh)_{ef}-h_{bc}\mathcal{L}_X(\nabla_dh)_{ef}\Big].
\end{align*}
Using the Leibniz rule for the Lie derivative (and that $X$ is Killing so $\mathcal{L}_XP=0$) gives
\begin{align*}
(\mathcal{J}^X[h])^a=2P^{abcdef}(\mathcal{L}_Xh)_{bc}(\nabla_dh)_{ef}-(\nabla_XY-\nabla_YX)^a,
\end{align*}
with $Y^a=P^{abcdef}h_{bc}\nabla_d h_{ef}$. Therefore, 
\begin{align*}
(\mathcal{J}^X[h])^a=2P^{abcdef}(\mathcal{L}_Xh)_{bc}(\nabla_dh)_{ef}-\nabla_h(X\otimes Y-Y\otimes X)^{ha}-(\mathrm{div}Y)X^a.
\end{align*}
where one uses that $\mathrm{div}X=0$ since $X$ is Killing. Now one can compute that
\begin{align*}
\mathrm{div}Y=P^{abcdef}\nabla_ah_{bc}\nabla_d h_{ef}+P^{abcdef}h_{bc}\nabla_a\nabla_d h_{ef}=P^{abcdef}\nabla_ah_{bc}\nabla_d h_{ef},
\end{align*}
where the last equality is by the linearised vacuum Einstein equation~(\ref{LVEE}). Now $X\otimes Y-Y\otimes X=A$ as defined in the proposition statement so then
\begin{align*}
(\nabla_a\nabla_hA)^{ha}=(\nabla_{[a}\nabla_{h]}A)^{ha}={R^h}_{bah}A^{ba}+{R^a}_{bah}A^{hb}=-2(\mathrm{Ric}(g))_{ab}A^{ab}=0.
\end{align*}
\end{proof}

\section{The Einstein Equation in Double Null Gauge}\label{DNG}

This section gives a brief introduction to the double null decomposition of a \underline{$4$-dimensional} spacetime following~\cite{FormBH} (see also the lecture notes~\cite{Aretakis}). In section~\ref{DNF} an introduction to double null canonical coordinates or `double null gauge' is presented. Some useful operations and calculations are introduced in section~\ref{DNO}. It proceeds with a double null decomposition of the Ricci coefficients and Weyl tensor in sections~\ref{DNR} and~\ref{DNW}. Section~\ref{DNS} discusses the double null foliations of the Schwarzschild spacetimes. Section~\ref{LinearDoubleNullGauge} discusses linearisation of the vacuum Einstein equation in the double null gauge around the Schwarzschild spacetime.
\subsection{Double Null Gauge for the Metric}\label{DNF}
A double null gauge is a set of coordinates $(u,v,\theta^A)$ on a manifold $\mathcal{M}$, with $A=1,2$, such that the metric takes the form
\begin{align}
g=-{{2}}\Omega^2(du\otimes dv+dv\otimes du)+\slashed{g}_{AB}(d\theta^A-b^Adv)\otimes (d\theta^B-b^Bdv).\label{DN}
\end{align}
The level sets of $u$ and $v$, denoted $C_u$ and $\underline{C}_v$ respectively, are null hypersurfaces and the sets $\mathcal{S}_{u,v}\doteq C_u\cap \underline{C}_v$ at fixed values of $(u,v)$ are homeomorphic to a $2$-sphere. One defines two normalised null vectors
\begin{align*}
e_3\doteq \frac{1}{\Omega}\partial_ u,\qquad e_4\doteq \frac{1}{\Omega}\big(\partial_v+b^Ae_A\big), \qquad 
\end{align*}
where $e_A\doteq \frac{\partial}{\partial\theta^A}\in \mathfrak{X}(\mathcal{S}_{u,v})$ for all $A=1,2$ and $u,v$. The symmetric $2$-tensor $\slashed{g}$ is the induced metric on $\mathcal{S}_{u,v}$. 

Associated submanifolds $\mathcal{S}_{u,v}$ are the notion of $\mathcal{S}_{u,v}$-tensors. These are defined as follows:
\begin{definition}[$\mathcal{S}_{u,v}$-tensors]
A vector field $X\in \mathfrak{X}(\mathcal{M})$ is an $\mathcal{S}_{u,v}$-vector field if
\begin{align}\label{tangent}
g(X,e_4)=0=g(X,e_3).
\end{align} Further, a $X\in \mathfrak{X}(\mathcal{S}_{u,v})$ can be viewed as a vector field in $\mathfrak{X}(\mathcal{M})$ satisfying~(\ref{tangent}). A one-form $\omega\in \Omega^1(M)$ is an $\mathcal{S}_{u,v}$-one-form if 
\begin{align}\label{tangent1f}
\omega(e_3)=0=\omega(e_4).
\end{align}
Similarly, one can view ${\omega}\in \Omega^1(\mathcal{S}_{u,v})$ as $\omega\in \Omega^1(M)$ satisfying~(\ref{tangent1f}). One extends these definitions naturally to arbitrary tensors. 
\subsection{Differential Operators on $\mathcal{S}_{u,v}$-Tensor Fields}\label{DNO}
In this section some useful operations on $\mathcal{S}_{u,v}$ tensors are defined. 
\begin{definition}[Operations]
For $\Theta_1$ and $\Theta_2$ be rank-$p$ $\mathcal{S}_{u,v}$ tensor fields, we define
\begin{align*}
\langle\Theta_1,\Theta_2\rangle=\slashed{g}^{A_1B_1}...\slashed{g}^{A_pB_p}(\Theta_1)_{A_1...A_p}(\Theta_2)_{B_1...B_p},\qquad |\Theta_1|^2=\langle\Theta_1,\Theta_1\rangle.
\end{align*}
For $\Phi$ be a $(0,2)$ $\mathcal{S}_{u,v}$-tensor field we define
\begin{align*}
\hat{\Phi}_{AB}&\doteq \frac{1}{2}\Big(\Phi_{AB}+\Phi_{BA}-(\mathrm{tr}\Phi )\slashed{g}_{AB}\Big).
\end{align*}
\end{definition}
\begin{definition}[Projected Covariant and Lie Derivatives]\label{Project34}
For $Y\in \mathfrak{X}(\mathcal{M})$, the projected covariant derivative $\ns_{Y}$ and the projected Lie derivative $\slashed{\mathcal{L}}_Y$ on a rank-$(0,p)$ $\mathcal{S}_{u,v}$-tensor field $T$ is defined as
\begin{align*}
(\ns_{Y}T)(X_1,...,X_p)&\doteq ({\nabla}_{Y}T)(X_1,...,X_p),\qquad (\slashed{\mathcal{L}}_{Y}T)(X_1,...,X_p)\doteq ({\mathcal{L}}_{Y}T)(X_1,...,X_p)
\end{align*}
for all $X_i\in \mathfrak{X}(\mathcal{S}_{u,v})$. In this work, the shorthand $\ns_{e_{\alpha}}=\ns_{\alpha}$ and $\slashed{\mathcal{L}}_{e_{\alpha}}=\slashed{\mathcal{L}}_{\alpha}$ for $\alpha=3,4,A$ will be adopted.

One defines $(p-1)$-covariant tensor fields $\divs T$ and $\slashed{\mathrm{curl}}T$ as
\begin{align*}
(\divs T)_{A_1...A_{p-1}}&\doteq\slashed{g}^{BC}(\ns_BT)_{CA_1...A_{p-1}},\qquad
(\slashed{\mathrm{curl}}T)_{A_1...A_{p-1}}\doteq\slashed{\varepsilon}^{BC}(\ns_BT)_{CA_1...A_{p-1}},
\end{align*}
where $\slashed{\varepsilon}$ is the induced volume form on $\mathcal{S}_{u,v}$. 
\end{definition}
\begin{remark}\label{Leibnizprojectder}
By the Leibniz rule for $\nabla$ one has the Leibniz rule for $\ns_{\alpha}$ for $\alpha=3,4,A$.
\end{remark}
\begin{definition}[Formal Adjoint Operators]
For a $\mathcal{S}_{u,v}^2$ one-form $\xi$ we define
\begin{align*}
\slashed{\mathcal{D}}_2^{\star}\xi&\doteq-\frac{1}{2}\Big[(\slashed{\nabla}_A\xi)_B+(\slashed{\nabla}_B\xi)_A-(\slashed{\mathrm{div}}\xi)\slashed{g}_{AB}\Big] \, ,
\end{align*}
which is the formal $L^2$-adjoint of $\divs$. For $f_1,f_2\in C^{\infty}(\mathcal{S}^2_{u,v})$ we define
\begin{align*}
    [\slashed{\mathcal{D}}_1^{\star}(f_1,f_2)]_A&\doteq-\slashed{\nabla}_Af_1+\slashed{\varepsilon}_{AB}\slashed{\nabla}^Bf_2,
\end{align*}
which is the formal $L^2$ adjoint of $\slashed{\mathcal{D}}_1$, which maps an $\mathcal{S}^2_{u,v}$-$1$-form to the pair of functions $(\divs\xi,\curls\xi)$.
\end{definition}

This section concludes with a collection of useful results.
\begin{lemma}[Useful Identities for $\mathcal{S}_{u,v}$]\label{Useful2DLemma}
Let $X\in \mathfrak{X}(\mathcal{S}_{u,v})$, $\xi\in \Omega^1\mathcal{S}_{u,v}$ and $\Phi,\Theta$ be symmetric-traceless $\mathcal{S}_{u,v}$ $2$-tensors. Then
\begin{align*}
    (\slashed{\mathcal{L}}_X\xi)_{A}&=(\ns_X\xi)_A-(\Dts \xi)_{AB}X^B+\frac{1}{2}(\divs X)\xi_A+\frac{1}{2}(\slashed{\mathrm{curl}}X)(\star\xi)_A,\\
    (\slashed{\mathcal{L}}_X\Theta)&=(\ns_X\Theta)-\langle\Dts X,\Theta\rangle\slashed{g}+(\divs X)\Theta+(\slashed{\mathrm{curl}}X)(\star\Theta).
\end{align*}
Further, the following integrated identities hold:
\begin{align*}
\int_{\mathcal{S}_{u,v}}|\ns\Theta|^2\slashed{\varepsilon}&=\int_{\mathcal{S}_{u,v}}\Big[ 2|\divs \Theta|^2-\slashed{\mathrm{Scal}}|{\Theta}|^2\Big]\slashed{\varepsilon},\qquad
\int_{\mathcal{S}_{u,v}}|\ns\xi|^2\slashed{\varepsilon}=\int_{\mathcal{S}_{u,v}}\Big[|\slashed{\mathrm{curl}}\xi|^2+|\divs \xi|^2-\frac{1}{2}\slashed{\mathrm{Scal}}|\xi|^2\Big]\slashed{\varepsilon}.
\end{align*}
\end{lemma}
\begin{proof}
Direct computations. See~\cite{MyThesis} for details. 
\end{proof}
\end{definition}\subsection{Null Decomposition of Ricci Coefficients}\label{DNR}
It is particularly convenient to decompose the Ricci coefficients in the normalised null frame. One makes the following definition:
\begin{definition}[Connection coefficients]\label{connectioncoeff}
Define the following $\mathcal{S}_{u,v}$-tensor fields:
\begin{equation}
    \begin{aligned}
       \chi_{AB}&\doteq g(\nabla_Ae_4,e_B),\\
       \eta_A&\doteq {\frac{1}{2}}g(\nabla_{3}e_4,e_A),\\
       \hat{\omega}&\doteq -{\frac{1}{2}}g(\nabla_4e_4,e_{3}),
    \end{aligned}
    \qquad
    \begin{aligned}
       \underline{\chi}_{AB}&\doteq g(\nabla_Ae_{3},e_B),\\
 \underline{\eta}_A&\doteq {\frac{1}{2}}g(\nabla_{4}e_{3},e_A),\\
 \underline{\hat{\omega}}&\doteq -{\frac{1}{2}}g(\nabla_{3}e_{3},e_{4})
    \end{aligned}
\end{equation}
and
\begin{align*}
\zeta_A&\doteq {\frac{1}{2}}g(\nabla_Ae_4,e_{3}).
\end{align*}
Extend these to tensor fields on $M$ by zero on $e_{3}$ and $e_4$. Additionally, it is useful to define
\begin{align*}
    \omega\doteq\Omega\hat{\omega},\qquad \underline{\omega}\doteq\Omega\hat{\underline{\omega}}.
\end{align*}
\end{definition}
\begin{remark}
One can check that the $\mathcal{S}_{u,v}$-tensors $\chi$ and $\underline{\chi}$ are symmetric. 
\end{remark}
It is conventional to decompose $\chi$ and $\underline{\chi}$ in the following manner:
\begin{definition}[Shear,Expansion]
The traceless part of $\chi$ is called the shear $\hat{\chi}$ and the trace of $\chi$ is called the expansion. Therefore,
\begin{align*}
\chi=\hat{\chi}+\frac{\mathrm{tr}\chi}{2}\slashed{g}.
\end{align*}
\end{definition}

\subsection{Null Decomposition of the Weyl Tensor}\label{DNW}
One can also decompose the Riemann tensor
\begin{align*}
    R_{\alpha\beta\gamma\delta}\doteq g(e_{\alpha},R(e_{\gamma},e_{\delta})e_{\beta}),
\end{align*}
with respect to the normalised null frame as follows.
\begin{definition}[Curvature Components]\label{curvcomp}
One defines the following $\mathcal{S}_{u,v}$ tensors:
\begin{equation*}
\begin{aligned}[c]
\alpha_{AB}&\doteq  R_{A4B4},\\
\beta_A&\doteq \frac{1}{2}R_{A434},\\
\rho& \doteq \frac{1}{4}R_{3434},
\end{aligned}
\qquad
\begin{aligned}[c]
\underline{\alpha}_{AB}&\doteq  R_{A3B3},\\
 \underline{\beta}_A &\doteq \frac{1}{2}R_{A334}\\
 \sigma&\doteq\frac{1}{4}(\star R)_{3434},
\end{aligned}
\end{equation*}
where $(\star R)_{abcd}\doteq \frac{1}{2}\varepsilon_{efab}{R^{ef}}_{cd}$ and $\varepsilon$ is the volume form for $(\mathcal{M},g)$. 
\end{definition}
\begin{remark}
Naively, the number of degrees of freedom encoded in ${W}\doteq (\alpha,\underline{\alpha},\beta,\underline{\beta},\rho,\sigma)$ is $12$. However, if $g$ satisfies the vacuum Einstein equation~\eqref{VE}, $\alpha$ and $\underline{\alpha}$ are trace-free and the null curvature components ${W}$ encode the $10$ Weyl degrees of freedom. All other components of the Riemann tensor are determined by $W$.
\end{remark}

\subsection{The Double Null Foliation of the $(\mathrm{Schw}_4,g_s)$ Exterior}\label{DNS}
The double null Eddington--Finkelstein coordinates $(u,v,\theta,\phi)$ provide a set of double null coordinates for the exterior region of the $4$-dimensional Schwarzschild black hole. The metric in these coordinates is given in equation~\eqref{SchTDN}. So one has
\begin{align*}
\Omega(u,v)^2=D(r(u,v)),\qquad b^A\equiv 0,\qquad \slashed{g}=r(u,v)^2\gamma_{2}.
\end{align*}
The normalised null frame is simply
\begin{align*}
e_3=\frac{1}{\Omega}\partial_u,\qquad e_4=\frac{1}{\Omega}\partial_v,\qquad e_A=\partial_{A}.
\end{align*}

One can calculate all Ricci coefficients and curvature components explicitly in terms of $r$. The only non-vanishing Ricci coefficients are
\begin{align*}
\quad{(\Omega\mathrm{tr}\chi)}=-{(\Omega\mathrm{tr}\underline{\chi})}=\frac{2\Omega^2}{r},\quad{{\omega}}=-{{\underline{\omega}}}=\frac{M}{r^{2}}
\end{align*}
and the only non-vanishing double null curvature component is
\begin{align*}
\quad{\rho}&=-\frac{2M}{r^{3}}.
\end{align*}
One has that the Riemann, Ricci and Scalar curvature of $\mathbb{S}^{2}_{u,v}$ are
\begin{align*}
\slashed{{R}}_{ABCD}&=\frac{1}{r^2}(\slashed{g}_{AC}\slashed{g}_{BD}-\slashed{g}_{AD}\slashed{g}_{BC}),\qquad \slashed{\mathrm{Ric}}=\frac{1}{r^2}\slashed{g},\qquad \slashed{\mathrm{Scal}}=\frac{2}{r^2}.
\end{align*}
Additionally, one has $\slashed{\mathrm{Scal}}=2\mathrm{K}$ where $\mathrm{K}$ is the Gauss curvature of $\mathbb{S}_{u,v}^2$. So, $\mathrm{K}=\frac{1}{r^2}$. \\

One can check using definition~\ref{connectioncoeff} the following proposition:
\begin{prop}\label{conncoeffSchw}
The Ricci coefficients for $(\mathrm{Schw}_4,g_s)$ in double null Eddington--Finkelstein coordinates satisfy the following relations:
\begin{align*}
\nabla_Ae_B&=\slashed{\Gamma}_{AB}^Ce_C+{\frac{1}{4}}\mathrm{tr}\chi(e_3-e_4)\slashed{g}_{AB}
\end{align*}
and
\begin{equation*}
    \begin{aligned}[c]
        \nabla_3e_A&=-\frac{\mathrm{tr}{\chi}}{2}e_A,\\
        \nabla_3e_4&={\hat{\omega}}e_4,
    \end{aligned}
    \qquad
    \begin{aligned}
       \nabla_{4}e_A&=\frac{\mathrm{tr}\chi }{2}e_A,\\
       \nabla_{4}e_3&=-\hat{\omega}e_3,
    \end{aligned}
    \qquad 
    \begin{aligned}
        \nabla_Ae_3&=-\frac{\mathrm{tr}\chi}{2} e_A,\\
        \nabla_3e_3&=-\hat{{\omega}}e_3,
    \end{aligned}
    \qquad 
    \begin{aligned}
       \nabla_Ae_4&=\frac{\mathrm{tr}\chi}{2} e_A,\\
       \nabla_4e_4&=\hat{\omega}e_4.
    \end{aligned}
\end{equation*}
\end{prop}
Additionally definition~\ref{Project34} combined with proposition~\ref{conncoeffSchw} gives
\begin{prop}[Projected Derivatives of $p$-covariant $\mathbb{S}^{2}_{u,v}$-Tensor Fields]\label{ProDTF2Sch}
Let $T$ be a $p$-covariant $\mathbb{S}^{2}_{u,v}$-tensor field. Then in double null Eddington--Finkelstein coordinates on the $4$-dimensional Schwarzschild exterior one has
\begin{align*}
(\ns_3T)_{A_1...A_p}&=\frac{1}{\Omega}\Big(\partial_u(T_{A_1...A_p})+\frac{p}{2}(\Omega\mathrm{tr}{\chi})T_{A_1...A_p}\Big),\\
(\ns_4T)_{A_1...A_p}&=\frac{1}{\Omega}\Big(\partial_v(T_{A_1...A_p})-\frac{p}{2}(\Omega\mathrm{tr}\chi)T_{A_1...A_p}\Big),\\
(\ns_AT)_{B_1...B_p}&=\partial_A(T_{B_1...B_p})-\sum_{i=1}^p\slashed{\Gamma}^C_{AB_i}T_{B_1...\hat{B}_iC...B_p}. 
\end{align*}
where $\hat{B}_i$ denotes removing the $i^{\mathrm{th}}$ index and replacing it by $C$.
\end{prop}
Finally this subsection concludes with the following commutation lemma
\begin{lemma}[Commutation Lemma]\label{comlem}
Let $T$ be a $p$-covariant $\mathbb{S}^{2}_{u,v}$-tensor field. Then in double null Eddington--Finkelstein coordinates on the $n$-dimensional Schwarzschild--Tangherlini exterior one has
\begin{align*}
(\ns_3\ns_BT-\ns_B\ns_3T)_{A_1...A_p}&=\frac{\mathrm{tr}\chi}{2} \ns_BT_{A_1...A_p},\\
(\ns_4\ns_BT-\ns_B\ns_4T)_{A_1...A_p}&=-\frac{\mathrm{tr}\chi }{2}\ns_BT_{A_1...A_p},\\
(\ns_3\ns_4T-\ns_4\ns_3T)_{A_1...A_p}&=\hat{\omega}(\ns_3T+\ns_4T)_{A_1...A_p}.
\end{align*}
\end{lemma}
\subsection{Linearisation in Double Null Gauge}\label{LinearDoubleNullGauge}
To find a definition of linearised metric in double null gauge, consider a one-parameter family of metrics $g(\epsilon)$ in double null canonical coordinates of the form
\begin{align*}
g(\epsilon)&=-{{2}}\Omega^2(\epsilon)(du\otimes dv+dv\otimes du)+b^A\slashed{g}_{AB}(\epsilon)(d\theta^B\otimes dv+dv\otimes d\theta^B)\\
&\nonumber\quad+b^A(\epsilon)b^B(\epsilon)\slashed{g}_{AB}(\epsilon)dv\otimes dv+\slashed{g}_{AB}(\epsilon)d\theta^A\otimes d\theta^B\nonumber
\end{align*}
where ${\Omega}(0)$, ${\slashed{g}}(0)$ and ${b}(0)$ are the background values for the spacetime one wants to linearise around. Therefore, one takes
\begin{align*}
\Omega(\epsilon)&={\Omega}(0)+\epsilon\Olino ,\qquad\slashed{g}(\epsilon)={\slashed{g}}(0)+\epsilon\slashed{h},\qquad b(\epsilon)={b}(0)+\epsilon\bmlin
\end{align*}
where the quantities with a superscript `$(1)$' denote linear perturbations. In general, expanding the the metric $g(\epsilon)$ to linear order as $g(\epsilon)=g(0)+\epsilon h+\mathcal{O}(\epsilon^2)$,
where $h$ is a symmetric $2$-tensor leads to the following definition a linearised metric $h$ to be in double null gauge:
\begin{definition}[Double Null Gauge for $h$]\label{DNGDefinition}
A solution $h$ to the linearised vacuum Einstein equation~(\ref{LVEE}) is said to be in double null gauge if there exists a function $\Olino :M\rightarrow \mathbb{R}$, a vector $\bmlin^A\in \mathfrak{X}(\mathcal{S}_{u,v})$ and a symmetric $\mathcal{S}_{u,v}$ $2$-tensor $\slashed{h}$ such that
\begin{align*}
h&=-\frac{4\Olino }{\Omega}(f^3\otimes f^4+f^4\otimes f^3)-\frac{\bmlin_A}{\Omega}(f^4\otimes f^A+ f^A\otimes f^4)+\slashed{h}_{AB}f^A\otimes f^B,
\end{align*}
where $(f^3,f^4,f^A)$ is the dual basis to $(e_3,e_4,e_A)$ for the background metric $g(0)$.
\end{definition}
\begin{remark}
The paper~\cite{DHR} uses the notation $\glin_{AB}$ for $\slashed{h}_{AB}$. 
\end{remark}

\subsubsection{The Linearised Equations of Gravity Around $(\mathrm{Schw}_4,g_s)$}\label{sec:LNSSchw}
The linearised vacuum Einstein equation~(\ref{LVEE}) can be encoded in a set of equations called the linearised null structure equations and the linearised Bianchi identities. The former result from linearising the Levi-Civita connection condition and the definition of the Riemann tensor \textit{assuming} that $g$ satisfies the vacuum Einstein equation~\eqref{VE}. More detail can be found in~\cite{DHR} and~\cite{MyThesis}. All linearised quantities are denoted with `$(1)$'.
\begin{remark}\label{LVEEtoLNS}
Henceforth, we will use the terminology that $h$ satisfies the {linearised} vacuum Einstein equation~\eqref{LVEE} in double null gauge with the linearised null structure equations and Bianchi identities being satisfied. 
\end{remark}

\begin{prop}[Linearised First Variation Formulas]\label{LinMet}
The linearised metric coefficients $\Olino $, $\bmlin$ and $\slashed{h}$ satisfy:
\begin{equation*}
    \begin{aligned}[c]
        \ns_3(\mathrm{tr}\slashed{h})&=\frac{2}{\Omega}\otxb ,\\
\ns_3\hat{\slashed{h}}_{AB}&=2\xblin _{AB},\\
 e_3\big(\Olin\big)&=\Omega^{-1}\olinb ,\\
 \partial_u\bmlin^A&=2\Omega^2(\elin -\eblin )^A,
    \end{aligned}
    \qquad
     \begin{aligned}[c]
        \ns_4(\mathrm{tr}\slashed{h})&=\frac{2}{\Omega}\Big(\otx -\divs \bmlin\Big)\\
        \ns_4\hat{\slashed{h}}_{AB}&=2\xlin _{AB}+\frac{2}{\Omega}(\Dts \bmlin)_{AB},\\
         e_4\big(\Olin\big)&=\Omega^{-1}\olin,\\
         2\ns\big(\Olin\big)&=(\elin +\eblin ).
    \end{aligned}
\end{equation*}
\end{prop}
\begin{prop}[Linearised Transversal Propagation Equations for Expansions]\label{LinPropExp}
The linearised expansions $\otx $ and  $\otxb $ satisfy
\begin{align*}
\ns_4\otxb &=2\Omega\Big[\divs \eblin +\Big(\rlin +2\Big(\frac{\Olino }{\Omega}\Big)\rho\Big)\Big]-\frac{\mathrm{tr}\chi}{2}\Big(\otxb -\otx \Big),\\
\ns_3\otx &=2\Omega\Big[\divs \elin +\Big(\rlin +2\Big(\frac{\Olino }{\Omega}\Big)\rho\Big)\Big]-\frac{\mathrm{tr}\chi}{2}\Big(\otxb -\otx \Big).
\end{align*}
\end{prop}
\begin{prop}[Linearised Raychauduri Equations]\label{LinRay}
The linearised expansions $\otx $ and  $\otxb $ satisfy
\begin{align*}
\ns_4\otx &=-\mathrm{tr}\chi\otx +2\olin\mathrm{tr}\chi+\frac{2}{\Omega}{\omega}\otx ,\qquad
\ns_3\accentset{(1)}{(\Omega\mathrm{tr}{\underline{\chi}})}=\mathrm{tr}{\chi}\accentset{(1)}{(\Omega\mathrm{tr}{\underline{\chi}})}-2\olinb \mathrm{tr}{\chi}-\frac{2}{\Omega}{{\omega}}\otxb .
\end{align*}
\end{prop}
\begin{prop}[Linearised Equations for the Shears]\label{LinShear}
The linearised shears $\xlin $ and  $\xblin $ satisfy
\begin{equation*}
    \begin{aligned}[c]
         \ns_4\xlin &=\Big(\hat{\omega}-\mathrm{tr}\chi\Big)\xlin -\alin ,\\
         \ns_4\xblin &=-2\Dts \eblin +\frac{\mathrm{tr}\chi}{2}\big(\xlin -\xblin \big)-\hat{\omega}\xblin ,
    \end{aligned}
    \qquad
    \begin{aligned}
        \ns_3\xblin &=-\Big(\hat{\omega}-\mathrm{tr}\chi\Big)\xblin -\ablin ,\\
    \ns_3\xlin &=-2\Dts \elin +\frac{\mathrm{tr}\chi}{2}\big(\xlin -\xblin \big)+\hat{\omega}\xlin .
    \end{aligned}
\end{equation*}
\end{prop}
\begin{prop}[Linearised Torsion Propagation Equations]\label{LinTorProp}
The linearised torsions $\elin $ and  $\accentset{(1)}{{\underline{\eta}}}$ satisfy
\begin{equation*}
    \begin{aligned}[c]
      \ns_4\elin &=-\blin -\frac{1}{2}\mathrm{tr}{\chi}(\elin -\eblin ),\\
      \ns_3\elin &=\frac{2}{\Omega}\ns\olinb +(\mathrm{tr}\chi)\elin -\bblin ,
    \end{aligned}
    \qquad 
    \begin{aligned}[c]
      \ns_3\eblin &=\bblin -\frac{1}{2}\mathrm{tr}{\chi}(\elin -\eblin ),\\
      \ns_4\eblin &=\frac{2}{\Omega}\ns\olin-(\mathrm{tr}\chi)\eblin +\blin .
    \end{aligned}
\end{equation*}
\end{prop}
\begin{prop}\label{Linomega}
The functions $\accentset{(1)}{\underline{{\omega}}}$ and $\accentset{(1)}{{{\omega}}}$ satisfy
\begin{align*}
\ns_4\accentset{(1)}{\underline{{\omega}}}&=-\Omega\Big(2\Big(\frac{\Olino }{\Omega}\Big)\rho+\rlin \Big),\qquad
\ns_3\olin=-\Omega\Big(2\Big(\frac{\Olino }{\Omega}\Big)\rho+\rlin \Big).
\end{align*}
\end{prop}
\begin{prop}[Linearised Torsion Constraints]\label{LinTorCon}
The linearised torsions $\elin $ and  $\accentset{(1)}{{\underline{\eta}}}$ satisfy
\begin{align*}
    \slashed{\mathrm{curl}}\eblin &=-\slin ,\qquad
\slashed{\mathrm{curl}}\elin =\slin .
\end{align*}
\end{prop}
\begin{prop}[Linearised Gauss Equations]\label{LinGauss}
The linearised scalar curvature satisfies
\begin{align*}
\accentset{(1)}{\slashed{\mathrm{Scal}}}&=-2\rlin -\frac{1}{2\Omega}\mathrm{tr}\chi\Big[\otxb -\otx \Big]-(\mathrm{tr}\chi)^2\Big(\frac{\Olino }{\Omega}\Big).
\end{align*}
\end{prop}
\begin{corollary}\label{Gausscor}
The linearised metric coefficient $\hat{\slashed{h}}$ satisfies
\begin{align*}
\divs \divs \hat{\slashed{h}}&=-2\rlin +\frac{1}{2\Omega}\mathrm{tr}\chi\Big[\otx -\otxb \Big]-(\mathrm{tr}\chi)^2\Big(\frac{\Olino }{\Omega}\Big)+\frac{1}{2}\slashed{\Delta}\mathrm{tr}\slashed{h}+\frac{1}{2}\slashed{\mathrm{Scal}}(\slashed{g})\mathrm{tr}\slashed{h}.
\end{align*}
\end{corollary}
\begin{prop}[Linearised Codazzi Constraints]\label{LinCod}
The linearised shears $\xlin $ and  $\xblin $ satisfy
\begin{align*}
\divs \xlin &=\frac{1}{2\Omega}\ns\otx -\frac{1}{2}\mathrm{tr}\chi\eblin -\blin ,\qquad
\divs \xblin =\frac{1}{2\Omega}\ns\otxb +\frac{1}{2}\mathrm{tr}{\chi}\elin +\bblin .
\end{align*}
\end{prop}
\begin{prop}\label{LinBianchi}
Suppose $h$ in double null gauge satisfies the linearised vacuum Einstein equation~(\ref{LVEE}). Then, the linearised null Weyl curvature components satisfy the following (null-decomposed) linearised Bianchi identities on $(\mathrm{Schw}_4,g_s)$:
\begingroup
\allowdisplaybreaks
\begin{equation*}
\begin{aligned}[c]
\ns_4\rlin &=-\frac{3}{2}\Big(\rlin \mathrm{tr}{\chi}+\frac{1}{\Omega}\rho\otx \Big)+\divs \blin ,\\
\ns_4\slin &=-\slashed{\mathrm{curl}}\blin -\frac{3}{2}(\mathrm{tr}\chi)\slin ,\\
\ns_4\blin &=\hat{\omega}\blin +\divs \alin -2(\mathrm{tr}\chi)\blin ,\\
\ns_4\bblin &=-3\rho\eblin +\slashed{\mathcal{D}}_1^{\star}(\rlin ,\slin )-\big(\mathrm{tr}\chi+\hat{\omega}\big)\bblin ,\\
\ns_4\ablin &=-\Big(2\hat{{\omega}}+\frac{1}{2}\mathrm{tr}\chi\Big)\ablin +2\slashed{\mathcal{D}}^{\star}_{2}\bblin -3\rho\xblin .
\end{aligned}
\qquad
\begin{aligned}[c]
\ns_3\rlin &=\frac{3}{2}\Big(\rlin \mathrm{tr}{\chi}-\frac{1}{\Omega}\rho\otxb \Big)-\divs \bblin ,\\
\ns_3\slin &=-\slashed{\mathrm{curl}}\bblin +\frac{3}{2}(\mathrm{tr}{\chi})\slin \\
\ns_3\blin &=3\rho\elin +\slashed{\mathcal{D}}_1^{\star}(-\rlin ,\slin )+\big(\mathrm{tr}{\chi}+\hat{{\omega}}\big)\blin ,\\
\ns_3\bblin &=\hat{\underline{\omega}}\bblin -\divs \ablin -2\mathrm{tr}\underline{\chi}\bblin ,\\
\ns_3\alin &=\Big(2\hat{{\omega}}+\frac{1}{2}(\mathrm{tr}{\chi})\Big)\alin -2\slashed{\mathcal{D}}^{\star}_{2}\blin -3\rho\xlin , \\
\end{aligned}
\end{equation*}
\endgroup
\end{prop}

\section{Canonical Energy in Double Null Gauge}\label{Results}
In this section, the proofs of theorems~\ref{thm:CanEntoGusCon},~\ref{thm:CEntoGConb} and~\ref{TH5} are given. In section~\ref{setup} the computation of the canonical energy in double null gauge is setup and summaried. Section~\ref{PreComp} collects some preliminary computations which will be useful in the proof of the theorems~\ref{thm:CanEntoGusCon}-\ref{TH5}. The intensive parts of the computations for the canonical energy in double null gauge are then given in sections~\ref{TMR},~\ref{betaconsproof} and~\ref{alphacons} as the proofs of theorems~\ref{thm:CanEntoGusCon},~\ref{thm:CEntoGConb} and~\ref{TH5}.
\subsection{The Setup and Summary}\label{setup}
We evaluate the $T$-canonical energy conservation law (see section~\ref{CanLVEE}) for a smooth solution~$h$ of the linearised vacuum Einstein equation~(\ref{LVEE}) on the characteristic rectangle depicted in figure~\ref{fig1}. This yields
\begin{align}
\mathcal{E}_{u_0}^T[h](v_0,v_1)+\mathcal{E}^T_{v_0}[h](u_0,u_1)=\mathcal{E}^T_{u_1}[h](v_0,v_1)+\mathcal{E}^T_{v_1}[h](u_0,u_1),\label{CanConSch}
\end{align}
where 
\begin{align}
\mathcal{E}_{u}^T[h](v_0,v_1)&\doteq \mathcal{E}^T_{C_u\cap\{ v_0\leq v\leq v_1\}}[h]=2\int_{v_0}^{v_1}\int_{\mathbb{S}^2_{u,v}}du(\mathcal{J}^T[h])\Omega^2dv\slashed{\varepsilon}=2\int_{v_0}^{v_1}\int_{\mathbb{S}^2_{u,v}}(\mathcal{J}^T[h])^3\Omega dv\slashed{\varepsilon},\label{CanonF1a}\\
\mathcal{E}^T_{v}[h](u_0,u_1)&\doteq \mathcal{E}^T_{\underline{C}_v\cap\{ u_0\leq u\leq u_1\}}[h]=2\int_{u_0}^{u_1}\int_{\mathbb{S}^2_{u,v}}dv(\mathcal{J}^T[h])\Omega^2du\slashed{\varepsilon}=2\int_{u_0}^{u_1}\int_{\mathbb{S}^2_{u,v}}(\mathcal{J}^T[h])^4\Omega du\slashed{\varepsilon},\label{CanonF1b}
\end{align}
and $\mathcal{J}^T[h]$ is the vector defined in equation~\eqref{XCEC}. 

Recall that the Schwarzschild black hole spacetime $(\mathrm{Schw}_4,g_s)$ has three additional Killing fields associated to the spherical symmetry of the spacetime. Let $\Omega_k$ be the Killing fields on the sphere~$\mathbb{S}^2$, i.e.
\begin{align}\label{SO3}
\Omega_1&=\partial_{\phi},\qquad
\Omega_2=\sin\phi\partial_{\theta}+\cot\theta\cos\phi\partial_{\phi},\qquad
\Omega_3=\cos\phi\partial_{\theta}-\cot\theta\sin\phi\partial_{\phi}.
\end{align}
As discussed in section~\ref{sec:HOCE}, one has a canonical energy conservation law for $\mathcal{L}_{{\Omega}_k}h$ for each $k=1,2,3$. In fact, the more appropriate conservation law is the sum of $\mathcal{E}^T[\mathcal{L}_{\Omega_k}h]$. Denote
\begin{align*}
    \slashed{\mathcal{E}}^T_u[h](v_0,v_1)\doteq \sum_{i=1}^3\mathcal{E}_{u}[\mathcal{L}_{\Omega_i}h](v_0,v_1),\qquad \slashed{\mathcal{E}}^T_u[h](v_0,v_1)\doteq \sum_{i=1}^3\mathcal{E}_{v}[\mathcal{L}_{\Omega_i}h](u_0,u_1).
\end{align*}
Then $\slashed{\mathcal{E}}^T$ satisfies
\begin{align}
\slashed{\mathcal{E}}^T_{u_0}[h](v_0,v_1)+\slashed{\mathcal{E}}^T_{v_0}[h](u_0,u_1)=\slashed{\mathcal{E}}^T_{u_1}[h](v_0,v_1)+\slashed{\mathcal{E}}^T_{v_1}[h](u_0,u_1).\label{CanConSch2}
\end{align}
One can write the terms in this conservation law~(\ref{CanConSch}) explicitly as
\begin{align*}
\slashed{\mathcal{E}}^T_{u}[h](v_0,v_1)&=2\sum_{i=1}^3\int_{v_0}^{v_1}\int_{\mathbb{S}^2}(\mathcal{J}[\mathcal{L}_{\Omega_i}h])^3\Omega dv\slashed{\varepsilon},\qquad
\slashed{\mathcal{E}}^T_{v}[h](u_0,u_1)=2\sum_{i=1}^3\int_{u_0}^{u_1}\int_{\mathbb{S}^2}(\mathcal{J}[\mathcal{L}_{\Omega_i}h])^4\Omega du\slashed{\varepsilon}.
\end{align*}
In the following sections the currents $\mathcal{J}^T[h]$, $\sum_k\mathcal{J}^T[\mathcal{L}_{{\Omega}_k}h]$ and $\mathcal{J}^T[\mathcal{L}_Th]$ are computed in explicitly. The reader should note that the proof of the statements in theorems~\ref{thm:CanEntoGusCon}-\ref{TH5} are extremely computationally involved. 

It is instructive for the reader to revisit the computation of the sympletic current~\eqref{eq:sympf} for the wave equation in the section~\ref{sec:IntroCL}. This computation illustrates the proofs of theorems~\ref{thm:CanEntoGusCon}-\ref{TH5} nicely. The key ideas are as follows:
\begin{enumerate}
    \item[(i)] The wave equation~(\ref{eq:WE}), in conjunction with integration by parts on $\mathbb{S}^2_{u,v}$ can be used to simplify the fluxes. In this case the wave equation can be used to remove $\partial_u\partial_v\Psi$ in exchange for first order derivative terms and an angular Laplacian of the solution, which, again, can be integrated by parts. 
    \item[(ii)] If one adds $\frac{1}{\Omega^2r^2}\partial_v(r^2\mathfrak{F})$ to $(\mathcal{J}^T)^u$ and subtracts $\frac{1}{\Omega^2r^2}\partial_u(r^2\mathfrak{F})$ from $(\mathcal{J}^T)^v$ for some $\mathfrak{F}$ then one maintains a conservation law on hypersurfaces since the terms on the spheres at the corners of the characteristic rectangle cancel. 
    \item[(iii)] There are second order derivatives of $\Psi$ that cannot be exchanged for first order derivative terms via the wave equation (as in point (i)). For example $\Psi\partial_v^2\Psi$ appears in $\mathcal{J}^T$. One can use point (ii) to remove these terms. This is precisely what allows one to identify $\mathfrak{F}=\mathcal{A}$. 
\end{enumerate}

With this discussion of the wave equation in hand, some intuition for why the main result in theorem~\ref{thm:CanEntoGusCon} is true can be given. First one should note that the Schwarzschild spacetime only has a limited number of symmetries so there cannot be arbitrarily many \textit{independent} conservation laws. This means that \textit{a priori} there is likely some relation between the canonical energy conservation law and Holzegel's conservation law. Further, observe that the linearised null structure equations of section~\ref{sec:LNSSchw} have the form
\begin{align*}
\nabla h&=\Gamlin ,\qquad
\nabla\Gamlin =\Gamma\Gamlin +\Wlin,
\end{align*}
where $\Gamma$ is the background Ricci coefficients and $\Gamlin $ is the linear perturbations to the Ricci coefficients and $\Wlin$ denotes the linearised Weyl curvature.
Therefore, the flux densities $\mathcal{J}^T[h]$ involved in the canonical energy of $h$ are of the schematic form
\begin{align}
\mathcal{J}^T[h]=\mathcal{L}_Th\cdot\nabla h-h\cdot \nabla\mathcal{L}_Th=\Gamlin \cdot\Gamlin +\Gamma h\cdot \Gamlin +h\cdot \Wlin.\label{FDE}
\end{align} 
It turns out that, in analogy with the computation for the symplectic current~\eqref{eq:sympf} for the wave equation in the section~\ref{sec:IntroCL}, by using~\underline{only} the linearised null structure equations~(\ref{sec:LNSSchw}), this last term involving curvature in equation~(\ref{FDE}) can be replaced (exactly like $\Psi\partial_u\partial_v\Psi$, $\Psi\partial_v^2\Psi$ and $\Psi\partial_u^2\Psi$ for the wave equation) by the boundary term $\pm\mathcal{A}$ (defined in theorem~\ref{thm:CanEntoGusCon}) on the spheres $\mathbb{S}^2_{u_0,v_0}$, $\mathbb{S}^2_{u_1,v_0}$, $\mathbb{S}^2_{u_0,v_1}$ and $\mathbb{S}^2_{u_1,v_1}$. 

The intuition behind theorem~\ref{thm:CEntoGConb} is the following. If $h$ in double null gauge solves the linearised vacuum Einstein equation~(\ref{LVEE}) then so does $\mathcal{L}_{\Omega_k}h$. So if one writes the conservation law~\eqref{MODCE2} in terms of $h$ then one can replace it everywhere with $\mathcal{L}_{\Omega_k}h$. Due to the identity $[T,\Omega_k]=0$ for all $k=1,2,3$, this operation commutes through each term in equations~(\ref{F1}) and~(\ref{F2}). Therefore, one can replace each linearised Ricci coefficient $\Gamlin $ with $\mathcal{L}_{\Omega_k}\Gamlin $. Roughly speaking, $\sum_k\slashed{\mathcal{L}}_{\Omega_k}\xlin $ is similar to the divergence operator $\divs $ on $\mathbb{S}^2_{u,v}$ acting on the linearised shear $\xlin $. From linearised Codazzi equations in proposition~\ref{LinCod} one can see that
\begin{align*}
\divs \xlin &=-\blin +\ldots,\qquad
\divs \xblin =\bblin +\ldots.
\end{align*}
Using the linearised null structure equations of propositions~\ref{LinMet}-\ref{LinCod} in section~\ref{sec:LNSSchw}, the linearised Bianchi equations of proposition~\ref{LinBianchi} and integration by parts one can then establish theorem~\ref{thm:CEntoGConb}.

Finally, the intuition behind theorem~\ref{TH5} is the following. Following the same reasoning as discussed above for theorem~\ref{thm:CEntoGConb} one can replace each metric coefficient $h$ and each linearised Ricci coefficient $\Gamlin $ in equations~(\ref{F1}) and~(\ref{F2}) with $\mathcal{L}_Th$ and $\mathcal{L}_{T}\Gamlin $, respectively. From linearised shear equations in proposition~\ref{LinShear} one can see that
\begin{align*}
\slashed{\mathcal{L}}_T\xlin &=\ns_3\xlin +\ns_4\xlin =-\alin +\ldots,\qquad
\slashed{\mathcal{L}}_T\xblin =\ns_3\xblin +\ns_4\xblin =-\ablin +\ldots.
\end{align*}
Using the linearised null structure equations of propositions~\ref{LinMet}-\ref{LinCod} in section~\ref{sec:LNSSchw}, the linearised Bianchi equations of proposition~\ref{LinBianchi} and integration by parts one can then establish theorem~\ref{TH5}.

\subsection{Preliminary Computations}\label{PreComp}
One should note that from the definition of double null Eddington--Finkelstein coordinates in section~\ref{DNS} one has the following relations
\begin{equation}\label{backgroundcomp}
    \begin{aligned}[c]
    \frac{1}{r^2}\ns_4r^2&=(\mathrm{tr}\chi),\\
        \ns_4\Omega&=\omega,\\
        \ns_4(\Omega\mathrm{tr}\chi)&=2\omega\mathrm{tr}\chi-\frac{\Omega}{2}(\mathrm{tr}\chi)^2,\\
        \ns_4\rho&=-\frac{3}{2}\rho\mathrm{tr}\chi,\\
        \ns_4\omega&=\Omega\rho,
    \end{aligned}
    \qquad
     \begin{aligned}[c]
     \frac{1}{r^2}\ns_3r^2&=-(\mathrm{tr}\chi),\\
        \ns_3\Omega&=-\omega,\\
        \ns_3(\Omega\mathrm{tr}\chi)&=\frac{\Omega}{2}(\mathrm{tr}\chi)^2-2\omega\mathrm{tr}\chi,\\
        \ns_3\rho&=\frac{3}{2}\rho\mathrm{tr}\chi,\\
        \ns_3\omega&=-\Omega\rho
    \end{aligned}
\end{equation}
and
\begin{align}\label{rhotoomega}
    \rho=-\hat{\omega}\mathrm{tr}\chi.
\end{align}
The relations~\eqref{backgroundcomp} and\eqref{rhotoomega} will be used liberally throughout the rest of this section and the next. 

For the canonical energy calculation one needs to compute $\nabla_{\alpha}h_{\beta\gamma}$ and $\nabla_{\alpha}(\mathcal{L}_Th)_{\beta\gamma}$ in the double null basis. The following lemma is useful for this. 
\begin{lemma}\label{lemmacompsch}
Let $\mathrm{S}$ be a symmetric $2$-tensor on the Schwarzschild black hole exterior $(\mathrm{Schw}_4,g_s)$ with
\begin{align*}
\mathrm{S}_{44}=\mathrm{S}_{33}=\mathrm{S}_{3A}=0,
\end{align*} 
in the normalised null basis $(e_3,e_4,\partial_{A})$ associated to the double null Eddington--Finkelstein coordinates. Further, denote $\mathrm{v}^{\mathrm{S}}_A\doteq \mathrm{S}_{4A}$ and $\slashed{\mathrm{S}}_{AB}\doteq\mathrm{S}_{AB}$ which are considered as the components of a $\mathbb{S}^2_{u,v}$-covector and symmetric $\mathbb{S}^2_{u,v}$ $2$-tensor. Then the \underline{non-zero} components of $(\nabla_{\alpha}\mathrm{S})_{\beta\gamma}$ have the following decomposition:
\begin{equation}
    \begin{aligned}[c]
        (\nabla_{3}\mathrm{S})_{43}&=e_3(\mathrm{S}_{43}),\\
        (\nabla_{4}\mathrm{S})_{A4}&=\ns_4\mathrm{v}^{\mathrm{S}}_A-\hat{\omega}\mathrm{v}^{\mathrm{S}}_A,\\
        (\nabla_A\mathrm{S})_{44}&=-(\mathrm{tr}\chi)\mathrm{v}^{\mathrm{S}}_A,\\
        (\nabla_4\mathrm{S})_{AB}&=(\ns_4\slashed{\mathrm{S}})_{AB},\\
        (\nabla_A\mathrm{S})_{3B}&=\frac{1}{2}(\mathrm{tr}\chi)\Big(\mathrm{S}_{AB}+\frac{1}{2}\mathrm{S}_{34}\slashed{g}_{AB}\Big),
    \end{aligned}
    \qquad
    \begin{aligned}[c]
        (\nabla_{4} \mathrm{S})_{43}&=e_4(\mathrm{S}_{43}),\\
        (\nabla_{3}\mathrm{S})_{A4}&=\ns_3\mathrm{v}^{\mathrm{S}}_A-\hat{\omega}\mathrm{v}^{\mathrm{S}}_A,\\
        (\nabla_A\mathrm{S})_{34}&=\partial_A(\mathrm{S}_{34})+\frac{1}{2}(\mathrm{tr}\chi)\mathrm{v}^{\mathrm{S}}_A,\\
        (\nabla_3\mathrm{S})_{AB}&=(\ns_3\slashed{\mathrm{S}})_{AB},\\
        (\nabla_A\mathrm{S})_{4B}&=\ns_A\mathrm{v}^{\mathrm{S}}_B-\frac{1}{2}(\mathrm{tr}\chi)\Big(\mathrm{S}_{AB}+\frac{1}{2}\mathrm{S}_{34}{\slashed{g}}_{AB}\Big)
    \end{aligned}
\end{equation}
and
\begin{align*}
(\nabla_A\mathrm{S})_{BC}&=(\ns_A\slashed{\mathrm{S}})_{BC}+\frac{1}{4}(\mathrm{tr}{\chi}){\slashed{g}}_{AC}\mathrm{v}^{\mathrm{S}}_{B}+\frac{1}{4}(\mathrm{tr}{\chi}){\slashed{g}}_{AB}\mathrm{v}^{\mathrm{S}}_C.
\end{align*}
\end{lemma}
\begin{proof}
One can calculate the above results using the formula
\begin{align*}
(\nabla_{\mu} \mathrm{S})_{\alpha\beta}=e_{\mu}( \mathrm{S}_{\alpha\beta})-\mathrm{S}(\nabla_{\mu}e_{\alpha},e_{\beta})-\mathrm{S}(e_{\alpha},\nabla_{\mu}e_{\beta}),
\end{align*}
in conjunction with the proposition~\ref{conncoeffSchw}.
\end{proof}
It will turn out that for calculating the canonical energy in double null gauge only the following non-zero components of $\nabla_{\alpha}h_{\beta\gamma}$ will be required:
\begin{prop}\label{hcalc4D}
Let $h$ be a smooth solution to the linearised vacuum Einstein equation~(\ref{LVEE}) in double null gauge on the $4$-dimensional Schwarzschild exterior. Then, in the normalised null basis $(e_3,e_4,\partial_{A})$ associated to the double null Eddington--Finkelstein coordinates, one has
\begin{equation}
    \begin{aligned}[c]
        (\nabla_{3} h)_{43}&=-\frac{4}{\Omega}\underline{\olin},\\
        (\nabla_{4} h)_{A4}&=-\frac{1}{\Omega}(\ns_4\bmlin)_A+2\hat{\omega}\frac{\bmlin_A}{\Omega},\\
        \widehat{(\nabla_4h)}_{AB}&=2\xlin _{AB}+\frac{2}{\Omega}(\Dts \bmlin)_{AB},
    \end{aligned}
    \qquad
    \begin{aligned}[c]
        (\nabla_{4} h)_{43}&=-\frac{4}{\Omega}{\olin},\\
        (\nabla_{3} h)_{A4}&=\frac{\bmlin_A}{2\Omega}(\mathrm{tr}\chi)-2(\elin -\eblin )_A,\\
        \widehat{(\nabla_3h)}_{AB}&=2\xblin _{AB}.
    \end{aligned}
\end{equation}
\end{prop}
\begin{proof}
To prove this statement note that for $h$ in double null gauge
\begin{align}\label{h}
h_{44}=h_{33}=h_{3A}=0,\quad h_{34}=-4\Big(\frac{\Olino }{\Omega}\Big),\quad  \mathrm{v}^h_B=h_{4B}=-\frac{\bmlin_B}{\Omega},\quad h_{AB}=\slashed{h}_{AB}.
\end{align}
The results follow directly from lemma~\ref{lemmacompsch} and the linearised null structure equations (in particular, proposition~\ref{LinMet}). The reader should note the decomposition
\begin{align*}
(\ns_3\slashed{h})_{AB}&=\ns_3(\mathrm{tr}\slashed{h})\slashed{g}_{AB}+(\ns_3\hat{\slashed{h}})_{AB},\qquad
(\ns_4\slashed{h})_{AB}=\ns_4(\mathrm{tr}\slashed{h})\slashed{g}_{AB}+(\ns_4\hat{\slashed{h}})_{AB}.
\end{align*}
\end{proof}
The following computation gives the components of $\mathcal{L}_Th$ in the normalised null frame.
\begin{prop}\label{lieder}
Let $T=\partial_t$ be the Killing field associated to stationarity of $(\mathrm{Schw}_4,g_s)$. Further, let $h$ a smooth solution to the linearised vacuum Einstein equation~(\ref{LVEE}) in double null gauge on the $4$-dimensional Schwarzschild exterior. Then, in the basis $(e_3,e_4,\partial_{A})$ associated to the double null Eddington--Finkelstein coordinates, $\mathcal{L}_Th$ has the following components
\begin{align*}
(\mathcal{L}_Th)_{44}&=0,\quad(\mathcal{L}_Th)_{33}=0,\quad (\mathcal{L}_Th)_{A3}=0,
\quad (\mathcal{L}_Th)_{34}=-2(\olin+\olinb ),\\  (\mathcal{L}_Th)_{4A}&=\mathrm{v}^{\mathcal{L}_Th}_A=-\frac{1}{2}(\ns_4\bmlin)_A-\Omega(\elin -\eblin )_A+\frac{1}{4}(\mathrm{tr}\chi)\bmlin_A,\\
(\mathcal{L}_Th)_{AB}&=(\slashed{\mathcal{L}}_T\slashed{h})_{AB}=\frac{1}{2}(\mathcal{L}_T\mathrm{tr}\slashed{h})\slashed{g}_{AB}+\Omega\xlin _{AB}+\Omega\xblin _{AB}+(\Dts \bmlin)_{AB}.
\end{align*}
\end{prop}
\begin{proof}
First note that since $t=u+v$ and $r_{\star}=v-u$ one has $T=\frac{\Omega}{2}(e_3+e_4)$. From this one can compute $\nabla_4T={\omega}e_4$, $\nabla_3T=-{\omega}e_3$ and $\nabla_AT=0$.  Also,
\begin{align*}
(\nabla_Th)_{\alpha\beta}=\frac{\Omega}{2}(\nabla_3h)_{\alpha\beta}+\frac{\Omega}{2}(\nabla_4h)_{\alpha\beta}.
\end{align*}
Hence, one can use proposition~\ref{hcalc4D} to compute $(\nabla_Th)_{\alpha\beta}$. Finally one can finish the calculations by using the formula
\begin{align*}
(\mathcal{L}_Th)_{\alpha\beta}=(\nabla_Th)_{\alpha\beta}+h_{\gamma\beta}(\nabla_{\alpha}T)^{\gamma}+h_{\gamma\alpha}(\nabla_{\beta}T)^{\gamma}
\end{align*}
in conjunction with lemma~\ref{lemmacompsch} for $h$ in double null gauge and proposition~\ref{hcalc4D}.
\end{proof}
\begin{prop}
Let $h$ be a smooth solution to the linearised vacuum Einstein equation~(\ref{LVEE}) in double null gauge on the $4$-dimensional Schwarzschild exterior expressed in double null Eddington--Finkelstein coordinates. Then, in the basis $(e_3,e_4,\partial_A)$,
\begin{align*}
(\nabla_{3}(\mathcal{L}_Th))_{43}&=2\Omega\Big(\rlin +2\rho\Big(\frac{\Olino }{\Omega}\Big)\Big)-\frac{2}{\Omega}\partial_u\olinb ,\\
(\nabla_{4} (\mathcal{L}_Th))_{43}&=-\frac{2}{\Omega}\partial_v\olin+2\Omega\Big(\rlin +2\rho\Big(\frac{\Olino }{\Omega}\Big)\Big),\\
(\nabla_{3} (\mathcal{L}_T h))_{A4}&=\frac{1}{4}(\mathrm{tr}\chi)\big(\ns_4\accentset{(1)}{b}-\frac{1}{2}(\mathrm{tr}\chi)\bmlin \big)_A+2\ns_A(\olin-\olinb )+2\Omega(\bblin +\blin )_A-\frac{\Omega\mathrm{tr}\chi}{2}(\elin +3\eblin )_A,\\
(\widehat{\nabla_3(\mathcal{L}_Th)})_{AB}&=(\widehat{\ns_3(\slashed{\mathcal{L}}_T\slashed{h})})_{AB}=\frac{1}{2}(\Omega\mathrm{tr}\chi)(\xblin +\xlin )-2\Omega \Dts \accentset{(1)}{{\underline{\eta}}}-2\omega\xblin -\Omega\ablin ,\\
(\widehat{\nabla_4(\mathcal{L}_Th)})_{AB}&=(\widehat{\ns_4(\slashed{\mathcal{L}}_T\slashed{h})})_{AB}=\Dts (\ns_4\bmlin )-\frac{\mathrm{tr}\chi}{2}\Dts \bmlin -\frac{\Omega\mathrm{tr}\chi}{2}(\xblin +\xlin )-2\Omega \Dts \eblin +2\omega\xlin -\Omega\alin .
\end{align*}
\end{prop}
\begin{proof}
To prove these relations one uses lemma~\ref{lemmacompsch} and proposition~\ref{lieder} to double null decompose the above quantities. The results above then follow from an application of the commutation lemma~\ref{comlem} and the linearised null structure equations of section~\ref{sec:LNSSchw}. 
\end{proof}
\begin{prop}\label{htracerel}
Let $h$ be a smooth solution to the linearised vacuum Einstein equation~(\ref{LVEE}) in double null gauge on the $4$-dimensional Schwarzschild exterior. Then one has the following relations
\begin{align*}
\slashed{\mathcal{L}}_T\mathrm{tr}\slashed{h}&=\otxb +\otx -\divs \bmlin ,\\
\ns_3(\slashed{\mathcal{L}}_T\mathrm{tr}\slashed{h})&=\frac{1}{\Omega}\Big(2\Omega^2\divs \eblin +2\Omega^2\Big(\rlin +\frac{\Olino }{\Omega}\rho\Big)+\frac{1}{2}(\Omega\mathrm{tr}\chi)\Big(\otxb +\otx \Big)\\
&\nonumber\quad-2\omega\otxb -2(\Omega\mathrm{tr}\chi)\olinb \Big),\nonumber\\
\ns_4(\slashed{\mathcal{L}}_T\mathrm{tr}\slashed{h})&=\frac{1}{\Omega}\Big(2\Omega^2\divs \accentset{(1)}{{\underline{\eta}}}+2\Omega^2\Big(\rlin +\frac{\Olino }{\Omega}\rho\Big)-\frac{1}{2}(\Omega\mathrm{tr}\chi)\Big(\otxb +\otx \Big)\\
&\nonumber\quad+2\omega\otx +2(\Omega\mathrm{tr}\chi)\olin-\Omega\divs (\ns_4\bmlin )+\frac{1}{2}(\Omega\mathrm{tr}\chi)\divs \bmlin \Big).\nonumber
\end{align*}
\end{prop}
\begin{proof}
The first equation follows from proposition~\ref{LinMet}. The rest of the results then follow from the linearised Raychauduri equations in proposition~\ref{LinRay} and propagation equations for the expansions in proposition~\ref{LinPropExp}. Note that for the last equation one uses the commutation lemma~\ref{comlem}.
\end{proof}

\subsection{Proof of Theorem~\ref{thm:CanEntoGusCon}}\label{TMR}
In the following subsection the main computation is performed. Many of the details are provided to leave the reader with no illusion as to the technical nature of the manipulations. 
\begin{proof}[Proof of theorem~\ref{thm:CanEntoGusCon}]
In this proof the following convention is adopted. The symbol $\equiv$ will denote equality under integration by parts on $\mathbb{S}^2_{u,v}$.

Recall that the $T$-canonical energy current $\mathcal{J}^T[h]^a$ is given by
\begin{align*}
\mathcal{J}^T[h]^a=&P^{abcdef}\Big[(\mathcal{L}_Th)_{bc}\nabla_dh_{ef}-h_{bc}\nabla_d(\mathcal{L}_Th)_{ef}\Big],
\end{align*}
with
\begin{align*}
P^{abcdef}\doteq g^{ae}g^{bf}g^{cd}-\frac{1}{2}g^{ad}g^{be}g^{cf}-\frac{1}{2}g^{ab}g^{ef}g^{cd}-\frac{1}{2}g^{ae}g^{df}g^{bc}+\frac{1}{2}g^{ad}g^{ef}g^{bc}.
\end{align*}
The inverse metric has a very simple form in the dual basis to the normalised null frame  $(e_3,e_4,\partial_{\theta},\partial_{\phi})$ associated to double null Eddington--Finkelstein coordinates. In particular, its \underline{non-zero} components are
\begin{align*}
    g^{34}=-\frac{1}{2},\qquad g^{AB}=\slashed{g}^{AB}.
\end{align*}
Recall that solution $h$ to the linearised vacuum Einstein equation~(\ref{LVEE}) in double null gauge is given by
\begin{align*}
h_{44}=0=h_{33}=h_{3A},\qquad h_{34}=-4\Big(\frac{\Olino }{\Omega}\Big),\qquad h_{4A}=-\frac{\bmlin _A}{\Omega},\qquad h_{AB}=\slashed{h}_{AB},
\end{align*}
in the basis $(e_3,e_4,\partial_{\theta},\partial_{\phi})$. Further by proposition~\ref{lieder}, the vanishing components of $\mathcal{L}_Th$ are
\begin{align*}
(\mathcal{L}_Th)_{44}&=0,\quad(\mathcal{L}_Th)_{33}=0,\quad (\mathcal{L}_Th)_{A3}=0.
\end{align*}
When calculating $\mathcal{J}^T[h]^4$ one should note that $\ns_3\slashed{g}=0$ and hence $[\mathrm{tr},\ns_3]=0$. From lemma~\ref{lemmacompsch} one has
\begin{align*}
    (\nabla_D\mathrm{S})_{3F}=\frac{1}{2}(\mathrm{tr}\chi)(\slashed{\mathrm{S}}_{DF}+\frac{1}{2}\mathrm{S}_{34}\slashed{g}_{DF})
\end{align*}
for $\mathrm{S}=h$ or $\mathrm{S}=\mathcal{L}_Th$. Hence,
\begin{align*}
\slashed{g}^{DF}(\nabla_D\mathrm{S})_{3F}=\frac{1}{2}(\mathrm{tr}\chi)(\mathrm{tr}\slashed{\mathrm{S}}+\mathrm{S}_{34}),
\end{align*}
for $\mathrm{S}=h$ or $\mathrm{S}=\mathcal{L}_Th$. Further, one can decompose $({\mathcal{L}}_Th)_{AB}$ into its trace and symmetric traceless part as
\begin{align*}
({\mathcal{L}}_Th)_{AB}&=\frac{1}{2}(\mathcal{L}_T\mathrm{tr}\slashed{h})\slashed{g}+\widehat{{\slashed{\mathcal{L}}}_T{h}}_{AB},\\
{\nabla}_3({\mathcal{L}}_Th)_{AB}&=\frac{1}{2}\ns_3(\mathcal{L}_T\mathrm{tr}\slashed{h})\slashed{g}+\widehat{\ns_3\slashed{\mathcal{L}}_Th}_{AB},
\end{align*}
where one uses that $(\mathcal{L}_Th)_{AB}=(\slashed{\mathcal{L}}_T\slashed{h})_{AB}$ and $\slashed{\mathcal{L}}_T\slashed{g}=0$. 
Combining these facts gives that $\mathcal{J}^T[h]^4$ can be written in a decomposed form as
\begin{align*}
\mathcal{J}^T[h]^4&=\frac{1}{4}\Big(\langle\widehat{\slashed{\mathcal{L}}_T\slashed{h}},\widehat{\ns_3\slashed{h}}\rangle-\langle\hat{ \slashed{h}},\widehat{\ns_3\slashed{\mathcal{L}}_T\slashed{h}}\rangle\Big)+\frac{1}{8}\Big(\mathcal{L}_T(\mathrm{tr}\slashed{h})e_3(h_{34})-(\mathrm{tr}\slashed{h})e_3((\mathcal{L}_Th)_{34})\Big)\nonumber\\
&\quad+\frac{1}{8}\Big(\ns_3(\mathcal{L}_T\mathrm{tr}\slashed{h})\mathrm{tr}\slashed{h}-(\mathcal{L}_T\mathrm{tr}\slashed{h})\ns_3(\mathrm{tr}\slashed{h})\Big)-\frac{1}{8}\ns_3\mathcal{L}_T(\mathrm{tr}\slashed{h})h_{34}\\
&\nonumber\quad+\frac{1}{8}\ns_3(\mathrm{tr}\slashed{h})(\mathcal{L}_Th)_{34}+\frac{1}{8}(\mathrm{tr}\chi)\Big(h_{34}\mathcal{L}_T(\mathrm{tr}\slashed{h})-(\mathcal{L}_Th)_{34}(\mathrm{tr}\slashed{h})\Big)\nonumber.
\end{align*}
Similarly, noting the relations derived above in lemma~\ref{lemmacompsch} and that
\begin{align*}
\frac{1}{4}\slashed{g}^{AB}\mathrm{v}^{\mathcal{L}_Th}_A\ns_B\mathrm{tr}\slashed{h}-\frac{1}{4}\slashed{g}^{AB}\mathrm{v}^{h}_A\ns_B\mathcal{L}_T\mathrm{tr}\slashed{h}&\equiv-\frac{1}{4}\divs \mathrm{v}^{\mathcal{L}_Th}\mathrm{tr}\slashed{h}+\frac{1}{4}\divs \mathrm{v}^{h}\mathcal{L}_T\mathrm{tr}\slashed{h},
\end{align*}
the component $\mathcal{J}^T[h]^3$ can be calculated as
\begin{align*}
\mathcal{J}^T[h]^3&\equiv\frac{1}{4}\Big(\langle\widehat{\slashed{\mathcal{L}}_T\slashed{h}},\widehat{\ns_4\slashed{h}}\rangle-\frac{1}{4}\langle \hat{\slashed{h}},\widehat{\ns_4(\slashed{\mathcal{L}}_T\slashed{h})}\rangle\Big)+\frac{1}{8}\Big(\ns_4(\mathcal{L}_T\mathrm{tr}\slashed{h})\mathrm{tr}\slashed{h}-(\mathcal{L}_T\mathrm{tr}\slashed{h})\ns_4(\mathrm{tr}\slashed{h})\Big)\nonumber\\
&\nonumber\quad+\frac{1}{8}\Big(\mathcal{L}_T(\mathrm{tr}\slashed{h})e_4(h_{34})-(\mathrm{tr}\slashed{h})e_4((\mathcal{L}_Th)_{34})-\ns_4\mathcal{L}_T(\mathrm{tr}\slashed{h})h_{34}+\ns_4(\mathrm{tr}\slashed{h})(\mathcal{L}_Th)_{34}\Big)\nonumber\\
&\quad-\frac{1}{8}(\mathrm{tr}\chi)\Big(h_{34}\mathcal{L}_T(\mathrm{tr}\slashed{h})-(\mathcal{L}_Th)_{34}(\mathrm{tr}\slashed{h})\Big)\\
&\nonumber\quad-\frac{1}{4}(\mathcal{L}_Th)_{34}\divs \mathrm{v}^h+\frac{1}{4}h_{34}\divs \mathrm{v}^{\mathcal{L}_Th}-\frac{1}{2}\langle\widehat{\slashed{\mathcal{L}}_Th},\widehat{\ns\mathrm{v}}^h\rangle+\frac{1}{2}\langle \hat{\slashed{h}},\widehat{\ns\mathrm{v}}^{\mathcal{L}_Th}\rangle\\
&\nonumber\quad+\frac{1}{4}\Big(\slashed{g}^{AB}\mathrm{v}^{\mathcal{L}_Th}_A(\nabla_3h)_{4B}-\slashed{g}^{AB}\mathrm{v}^{h}_A(\nabla_3\mathcal{L}_Th)_{4B}-\divs \mathrm{v}^{\mathcal{L}_Th}\mathrm{tr}\slashed{h}+\divs \mathrm{v}^{h}\mathcal{L}_T\mathrm{tr}\slashed{h}\Big),
\end{align*}
where $\mathrm{v}^{\mathrm{S}}_B=\mathrm{S}_{4B}$ (for $\mathrm{S}=h$ or $\mathrm{S}=\mathcal{L}_Th$) is considered as a covector. 

Further, using proposition~\ref{LinMet}, one has
\begin{align*}
\ns_A\mathrm{v}^h_B&=-\frac{1}{\Omega}\ns_Ab_B,\qquad
\ns_A\mathrm{v}^{\mathcal{L}_Th}_B=-\frac{1}{2}\ns_A(\ns_4\bmlin )_B-\Omega\ns_A(\elin -\eblin )_B+\frac{1}{2}\mathrm{tr}\chi\ns_A\bmlin _B.
\end{align*}
So,
\begin{align*}
\widehat{\ns\mathrm{v}^h}&=\frac{1}{\Omega}\Dts \bmlin ,\qquad 
\widehat{\ns\mathrm{v}^{\mathcal{L}_Th}}=\frac{1}{2}\Dts (\ns_4\bmlin )+\Omega(\Dts \elin -\Dts \eblin )-\frac{1}{2}\mathrm{tr}\chi\Dts \bmlin ,\\
\divs \mathrm{v}^h&=-\frac{1}{\Omega}\divs \bmlin ,\qquad \divs \mathrm{v}^{\mathcal{L}_Th}=-\frac{1}{2}\divs (\ns_4\bmlin )-\Omega\divs \elin +\Omega\divs \eblin +\frac{1}{2}\mathrm{tr}\chi\divs \bmlin .
\end{align*}
Therefore, exploiting these relations and propositions~\ref{hcalc4D}-\ref{htracerel}, one can write two complicated expressions for $\mathcal{J}^T[h]^4$  and $\mathcal{J}^T[h]^3$:
\begin{align}
\mathcal{J}^T[h]^4&\equiv\frac{\Omega}{2}|\xblin |^2-\frac{\omega}{\Omega}\Big(\frac{\Olino }{\Omega}\Big)\otxb -\frac{1}{2\Omega}\olinb \otx -\frac{1}{4\Omega}\otxb ^2+\frac{\Omega}{2}\langle\xlin ,\xblin \rangle\label{DirectCurrent1}\\
&\nonumber\quad-\frac{\Omega}{2}\langle\eblin ,\elin +\eblin \rangle-\frac{1}{4\Omega}\otx \otxb -\frac{1}{2\Omega}\olinb \Big(2\otxb -\divs \bmlin \Big)-\frac{1}{2\Omega}\olin\otxb \\
&\nonumber\quad+\frac{\Omega}{2}\langle\divs \hat{\slashed{h}},\eblin \rangle+\frac{\Omega}{4}\langle \hat{\slashed{h}}, \ablin \rangle+\frac{1}{2}\bblin (\bmlin )-\frac{(\Omega\mathrm{tr}\chi)}{8}\langle\xblin +\xlin ,\hat{\slashed{h}}\rangle+\frac{1}{2}\omega\langle \hat{\slashed{h}},\xblin \rangle-\frac{1}{4}\mathrm{tr}\chi\eblin (\bmlin )\\
&\nonumber\quad+\frac{\mathrm{tr}\chi}{16}\Big[\otxb +\otx \Big]\Big(\mathrm{tr}\slashed{h}-4\Big(\frac{\Olino }{\Omega}\Big)\Big)-\frac{1}{4\Omega}\Big(\frac{\Olino }{\Omega}\Big)\Big(4\partial_u\olin+4(\Omega\mathrm{tr}\chi)\olinb \Big)\\
&\nonumber\quad+\frac{1}{16\Omega}\Big(4\Omega^2\divs \eblin -4\omega\otxb +4(\partial_u\olinb )+4(\Omega\mathrm{tr}\chi)\olin\Big)\mathrm{tr}\slashed{h},
\end{align}
\begin{align}
\mathcal{J}^T[h]^3&\equiv\frac{\Omega}{2}|\xlin |^2-\frac{1}{4\Omega}\otx ^2-\frac{1}{2\Omega}\otxb \olin+\frac{\omega}{\Omega}\Big(\frac{\Olino }{\Omega}\Big)\otx +\frac{\Omega}{2}\langle\xblin ,\xlin \rangle\label{DirectCurrent2}\\
&\nonumber\quad+\frac{\Omega}{2}|\eblin |^2+\frac{1}{8}\mathrm{tr}\chi\langle\bmlin ,\elin +\eblin \rangle+\frac{1}{4}\langle\bmlin ,\bblin -\blin \rangle+\frac{1}{8}(\Omega\mathrm{tr}\chi)\langle \hat{\slashed{h}},(\xblin +\xlin )\rangle_{\slashed{g}}-\frac{3}{2}\Omega\langle\elin ,\eblin \rangle\\
&\nonumber\quad-\frac{1}{4\Omega}\otxb \otx -\frac{1}{\Omega}\otx \olin-\frac{1}{2\Omega}\otx \olinb +\frac{1}{2\Omega}\divs \bmlin \olin\\
&\nonumber\quad-\frac{1}{2}\omega\langle \hat{\slashed{h}},\xlin \rangle+\frac{\Omega}{4}\langle \hat{\slashed{h}},\alin \rangle+\frac{\Omega}{2}\langle\divs \hat{\slashed{h}},\elin \rangle+\frac{1}{4}\Big(\langle (\ns_4\bmlin ),(\elin -\eblin )\rangle+\langle\bmlin ,\ns_4\eblin \rangle-\langle\bmlin ,\ns_3\elin \rangle\Big)\\
&\nonumber\quad+\frac{1}{\Omega}\Big(\frac{\Olino }{\Omega}\Big)\Big({(\Omega\mathrm{tr}\chi)}\olin-\partial_u\olin\Big)+\frac{\Omega\mathrm{tr}\chi}{16}\Big[\otx +\otxb \Big]\Big(4\Big(\frac{\Olino }{\Omega}\Big)-\mathrm{tr}\slashed{h}\Big)\\
&\nonumber\quad+\frac{1}{4\Omega}\Big(\Omega^2\divs \elin +\partial_v\olin+\omega\otx -(\Omega\mathrm{tr}\chi)\olinb \Big)\mathrm{tr}\slashed{h},
\end{align}
where the following relation has been employed
\begin{align*}
    2\Big(\frac{\Olino }{\Omega}\Big)\divs \elin \equiv-\langle\elin ,\elin +\eblin \rangle
\end{align*}
and similarly for $\eblin $. Also, note that in simplifying these expressions one uses the linearised Codazzi equations in proposition~\ref{LinCod} to give that
\begin{align*}
\langle\Dts \bmlin ,2\xblin \rangle&\equiv 2\langle\bmlin ,\divs \xblin \rangle=\mathrm{tr}\chi\langle \bmlin ,\elin \rangle+2\langle \bmlin ,\bblin \rangle-\frac{1}{\Omega}\otxb \divs \bmlin ,\\
\langle\Dts \bmlin ,2\xlin \rangle&\equiv 2\langle\bmlin ,\divs \xlin \rangle=-\mathrm{tr}\chi\langle \bmlin ,\eblin \rangle-2\langle \bmlin ,\blin \rangle-\frac{1}{\Omega}\otx \divs \bmlin .
\end{align*}
The function $\mathcal{A}$ can be identified by the following observations: 
\begin{enumerate}
\item[(1)] If one considers the wider problem of interest, namely the conservation law on the boundary of a characteristic rectangle on the exterior of $(\mathrm{Schw}_4,g_s)$ then if one can write
\begin{align}
\mathcal{J}^T[h]^3&=\overline{\mathcal{J}^T[h]}^3-\frac{1}{r^2}\ns_4\mathfrak{F},
\qquad \mathcal{J}^T[h]^4=\overline{\mathcal{J}^T[h]}^4+\frac{1}{r^2}\ns_3\mathfrak{F},\label{Usefulremarkboundary}
\end{align} 
for some function $\mathfrak{F}$, then one has a cancellation of $\mathfrak{F}$ at the spheres at the corners of the characteristic rectangle.
\item[(2)] There are terms in $\mathcal{J}^4$ and $\mathcal{J}^3$ which appear with the correct derivative ($\partial_u$ for $\mathcal{J}^4$ and $\partial_v$ for $\mathcal{J}^3$) to integrate by parts but one has no expression for them in terms of the null structure equations. For example $\partial_u\olinb \mathrm{tr}\slashed{h}$ and $\langle\ns_4\bmlin ,(\elin -\elin )\rangle$. However, one has expressions for $\partial_v\olinb $ and $\partial_u\bmlin ^A$ from propositions~\ref{Linomega} and~\ref{LinMet} respectively. So, with point (1) in mind, expressing such terms as total derivative and adding and subtracting such terms is advantageous to manipulate the expressions for $\mathcal{J}^3$ and $\mathcal{J}^4$. 
\end{enumerate}
The function $\mathcal{A}$ (defined in theorem~\ref{thm:CanEntoGusCon}) can be written as
\begin{align*}
\mathcal{A}[h]&=\frac{1}{r^2}\Big(\mathcal{A}_1-\mathcal{A}_2-\mathcal{A}_3-\mathcal{A}_4+\mathcal{A}_5-\mathcal{A}_6-\mathcal{A}_7+\mathcal{A}_8+\mathcal{A}_9-\mathcal{A}_{10}\Big)
\end{align*}
with
\begin{equation}
    \begin{aligned}[c]
    \mathcal{A}_1&\doteq \frac{1}{4}r^2\olinb \mathrm{tr}\slashed{h},\\
    \mathcal{A}_4&\doteq \frac{r^2}{8}\otx \mathrm{tr}\slashed{h},\\
    \mathcal{A}_7&\doteq\frac{3}{2}r^2\Big(\frac{\Olino }{\Omega}\Big)\otxb ,
    \end{aligned}
    \qquad 
    \begin{aligned}[c]
    \mathcal{A}_2&\doteq \frac{1}{4}r^2\olin\mathrm{tr}\slashed{h},\\
    \mathcal{A}_5&\doteq \frac{r^2}{8}\otxb \mathrm{tr}\slashed{h},\\
    \mathcal{A}_8&\doteq \frac{3r^2}{2}\Big(\frac{\Olino }{\Omega}\Big)\otx ,
    \end{aligned}
    \qquad 
    \begin{aligned}[c]
        \mathcal{A}_3&\doteq \frac{r^2}{4}\langle\bmlin ,\elin -\eblin \rangle,\\
        \mathcal{A}_6&\doteq \frac{r^2\Omega}{4}\langle\xblin -\xlin ,\hat{\slashed{h}}\rangle,\\
        \mathcal{A}_9&\doteq \frac{r^2(\Omega\mathrm{tr}\chi)}{2}\Big(\frac{\Olino }{\Omega}\Big)\mathrm{tr}\slashed{h},
    \end{aligned}
\end{equation}
and
\begin{align*}
     \mathcal{A}_{10}&\doteq 2r^2{(\Omega\mathrm{tr}{\chi})}\Big(\frac{\Olino }{\Omega}\Big)^2.
\end{align*}
One can check that using proposition~\ref{LinMet} that
\begin{align*}
\frac{1}{r^2}\ns_3\mathcal{A}_1&=-\frac{1}{4}(\mathrm{tr}\chi)\olinb \mathrm{tr}\slashed{h}+\frac{1}{4\Omega}\partial_u\olinb \mathrm{tr}\slashed{h}+\frac{1}{2\Omega}\olinb \otxb ,\\
\frac{1}{r^2}\ns_4\mathcal{A}_1&=\frac{1}{4}\Big((\mathrm{tr}\chi)\olinb +\frac{1}{\Omega}\partial_v\olinb \Big)\mathrm{tr}\slashed{h}+\frac{\olinb }{2\Omega}\Big(\otx -\divs \bmlin \Big),\\
\frac{1}{r^2}\ns_3\mathcal{A}_2&=-\frac{1}{4}(\mathrm{tr}\chi)\olin\mathrm{tr}\slashed{h}+\frac{1}{4\Omega}\partial_u\olin\mathrm{tr}\slashed{h}+\frac{1}{2\Omega}\olin\otxb ,\\
\frac{1}{r^2}\ns_4\mathcal{A}_2&=\frac{1}{4}\Big((\mathrm{tr}\chi)\olin+\frac{1}{\Omega}\partial_v\olin\Big)\mathrm{tr}\slashed{h}+\frac{\olin}{2\Omega}\Big(\otx -\divs \bmlin \Big).
\end{align*}
Using propositions~\ref{LinMet} and~\ref{LinTorProp} one has
\begin{align*}
\frac{1}{r^2}\ns_3\mathcal{A}_3&=-\frac{1}{4}(\mathrm{tr}\chi)\langle\bmlin ,\elin -\eblin \rangle+\frac{\Omega}{2}|\elin -\eblin |^2+\frac{1}{4}\langle\bmlin ,\ns_3\elin \rangle-\frac{1}{4}\langle\bmlin ,\bblin \rangle,\\
\frac{1}{r^2}\ns_4\mathcal{A}_3&=\frac{\mathrm{tr}\chi}{8}(\elin -\eblin )(\bmlin )+\frac{1}{4}\Big[(\elin -\eblin )(\ns_4\bmlin )-\langle\bmlin ,\ns_4\eblin \rangle-\blin (\bmlin )\Big].
\end{align*}
Using propositions~\ref{LinMet},~\ref{LinPropExp} and~\ref{LinRay} gives
\begin{align*}
\frac{1}{r^2}\ns_3\mathcal{A}_4&=\frac{1}{4\Omega}\otxb \otx +\frac{\Omega}{4}\mathrm{tr}\slashed{h}\divs \elin -\frac{1}{4\Omega}(\partial_u\olin)\mathrm{tr}\slashed{h}\\
&\nonumber\quad-\frac{1}{16\Omega}\mathrm{tr}\slashed{h}(\Omega\mathrm{tr}\chi)\Big(\otxb +\otx \Big),\nonumber\\
\frac{1}{r^2}\ns_4\mathcal{A}_4&=\frac{1}{4\Omega}\otx ^2+\frac{1}{4\Omega}\mathrm{tr}\slashed{h}\Big(\omega\otx +(\Omega\mathrm{tr}\chi)\olin\Big)-\frac{1}{4\Omega}\otx \divs \bmlin 
\end{align*}
and
\begin{align*}
\frac{1}{r^2}\ns_3\mathcal{A}_5&=\frac{1}{4\Omega}\otxb ^2-\frac{1}{4\Omega}\mathrm{tr}\slashed{h}\Big(\omega\otxb +(\Omega\mathrm{tr}\chi)\olinb \Big),\\
\frac{1}{r^2}\ns_4\mathcal{A}_5&=\frac{1}{4\Omega}\otxb \Big(\otx -\divs \bmlin \Big)-\frac{1}{4\Omega}(\partial_v\olinb )\mathrm{tr}\slashed{h}\\
&\nonumber\quad+\frac{\mathrm{tr}\chi}{16}\mathrm{tr}\slashed{h}\Big[\otxb +\otx \Big]+\frac{\Omega}{4}\divs \eblin \mathrm{tr}\slashed{h}.\nonumber
\end{align*}
Using propositions~\ref{LinMet} and~\ref{LinShear} gives
\begin{align*}
\frac{1}{r^2}\ns_3\mathcal{A}_6&=\frac{1}{8}(\Omega\mathrm{tr}\chi)\langle\xblin +\xlin ,\hat{\slashed{h}}\rangle-\frac{\Omega}{2}\langle\xblin ,\xlin \rangle+\frac{\Omega}{2}|\xblin |^2+\frac{\Omega}{2}\langle\eblin ,\divs \hat{\slashed{h}}\rangle-\frac{1}{2}\omega\langle \xblin ,\hat{\slashed{h}}\rangle-\frac{\Omega}{4}\langle \ablin ,\hat{\slashed{h}}\rangle,\\
\frac{1}{r^2}\ns_4\mathcal{A}_6&\equiv\frac{\Omega\mathrm{tr}\chi}{8}\langle\xblin +\xlin ,\hat{\slashed{h}}\rangle+\frac{\Omega}{2}\langle\xblin ,\xlin \rangle-\frac{\Omega}{2}|\xlin |^2-\frac{\Omega}{2}\langle\eblin ,\divs \hat{\slashed{h}}\rangle-\frac{1}{2}\omega\langle \xlin ,\hat{\slashed{h}}\rangle+\frac{\Omega}{4}\langle \alin ,\hat{\slashed{h}}\rangle\\
&\nonumber\quad+\frac{1}{2}(\bblin +\blin )(\bmlin )+\frac{\mathrm{tr}\chi}{4}(\elin +\eblin )(\bmlin )-\frac{1}{4\Omega}\Big(\otxb -\otx \Big)\divs \bmlin .\nonumber
\end{align*}
With propositions~\ref{LinMet},~\ref{LinPropExp} and~\ref{LinRay} one has
\begin{align*}
\frac{1}{r^2}\ns_3\mathcal{A}_7&=\frac{3}{2\Omega}\olinb \otxb -\frac{3}{\Omega}\Big(\frac{\Olino }{\Omega}\Big)\Big(\omega\otxb +(\Omega\mathrm{tr}\chi)\olinb \Big),\\
\frac{1}{r^2}\ns_4\mathcal{A}_7&\equiv\frac{3}{2\Omega}\olin\otxb -\frac{3}{2}\Omega\langle\elin ,\eblin \rangle-\frac{3\Omega}{2}\Omega|\eblin |^2-\frac{3}{\Omega}\Big(\frac{\Olino }{\Omega}\Big)\partial_v\olinb \\
&\nonumber\quad+\frac{3}{4\Omega}\Big(\frac{\Olino }{\Omega}\Big)(\Omega\mathrm{tr}\chi)\Big(\accentset{(1)}{(\Omega\mathrm{tr}{\underline{\chi}})}+\otx \Big).\nonumber
\end{align*}
Analogously, with propositions~\ref{LinMet},~\ref{LinPropExp} and~\ref{LinRay} one has
\begin{align*}
\frac{1}{r^2}\ns_3\mathcal{A}_8&=\frac{3}{2\Omega}\olinb \otx -\frac{3}{2}\Omega|\elin |^2-\frac{3}{2}\Omega\langle\elin ,\eblin \rangle-\frac{3}{\Omega}\Big(\frac{\Olino }{\Omega}\Big)\partial_u\olin\\
&\nonumber\quad-\frac{3}{4}\Big(\frac{\Olino }{\Omega}\Big)(\mathrm{tr}\chi)\Big(\otx +\accentset{(1)}{(\Omega\mathrm{tr}{\underline{\chi}})}\Big),\\
\frac{1}{r^2}\ns_4\mathcal{A}_8&=\frac{3}{2\Omega}\olin\otx +\frac{3}{\Omega}\Big(\frac{\Olino }{\Omega}\Big)\Big(\omega\otx +(\Omega\mathrm{tr}\chi)\olin\Big).
\end{align*}
Using proposition~\ref{LinMet} gives
\begin{align*}
\frac{1}{r^2}\ns_3\mathcal{A}_9&=-\frac{\Omega}{4}(\mathrm{tr}\chi)^2\frac{\Olino }{\Omega}\mathrm{tr}\slashed{h}-\omega\mathrm{tr}\chi\frac{\Olino }{\Omega}\mathrm{tr}\slashed{h}+\frac{1}{\Omega}(\Omega\mathrm{tr}\chi)\frac{\Olino }{\Omega}\otxb +\frac{1}{2}(\mathrm{tr}\chi)\olinb \mathrm{tr}\slashed{h},\\
\frac{1}{r^2}\ns_4\mathcal{A}_9&=\frac{\mathrm{tr}\chi}{2}\olin\mathrm{tr}\slashed{h}+\frac{\mathrm{tr}\chi}{4}\Big[\Omega\mathrm{tr}\chi+4\omega\Big]\Big(\frac{\Olino }{\Omega}\Big)\mathrm{tr}\slashed{h}+\mathrm{tr}\chi\Big(\frac{\Olino }{\Omega}\Big)\otx +\frac{\mathrm{tr}\chi}{2}(\elin +\eblin )(\bmlin ),\\
\frac{1}{r^2}\ns_3\mathcal{A}_{10}&=-\Omega(\mathrm{tr}\chi)^2\Big(\frac{\Olino }{\Omega}\Big)^2-4\omega(\mathrm{tr}\chi)\Big(\frac{\Olino }{\Omega}\Big)^2+4{(\Omega\mathrm{tr}{\chi})}\Big(\frac{\Olino }{\Omega}\Big)\olinb ,\\
\frac{1}{r^2}\ns_4\mathcal{A}_{10}&=\mathrm{tr}\chi\Big[\Omega\mathrm{tr}\chi+4\omega\Big]\Big(\frac{\Olino }{\Omega}\Big)^2+4{(\Omega\mathrm{tr}{\chi})}\Big(\frac{\Olino }{\Omega}\Big)\olin.
\end{align*}
Therefore, denoting $(\overline{\mathcal{J}}^T[h])^3=(\mathcal{J}^T[h])^3+\frac{1}{r^2}\ns_4(r^2\mathcal{A})$ and $(\overline{\mathcal{J}}^T[h])^4=(\mathcal{J}^T[h])^4-\frac{1}{r^2}\ns_3(r^2\mathcal{A})$, one can calculate that
\begin{align*}
(\overline{\mathcal{J}}^T[h])^3&\equiv\Omega|\xlin |^2-\frac{1}{2\Omega}\otx ^2-\frac{2}{\Omega}\otxb \olin+\frac{4}{\Omega}\Big(\frac{\Olino }{\Omega}\Big)\omega\otx +2\Omega|\eblin |^2\\
&\nonumber\quad+\Big(\frac{\Olino }{\Omega}\Big)\Big\{\frac{2}{\Omega}\partial_v\olinb -\Omega\Big(\divs \divs \hat{\slashed{h}}-\frac{1}{2}\slashed{\Delta}\mathrm{tr}\slashed{h}\Big)-\frac{\mathrm{tr}\chi}{2}\Big(\accentset{(1)}{(\Omega\mathrm{tr}{\underline{\chi}})}-\otx \Big)\\
&\nonumber\quad+\frac{\mathrm{tr}\chi}{4}\Big(\Omega\mathrm{tr}\chi+4\omega\Big)\Big(\mathrm{tr}\slashed{h}-4\Big(\frac{\Olino }{\Omega}\Big)\Big)\Big\},\\
(\overline{\mathcal{J}}^T[h])^3&\equiv\Omega|\xblin |^2-\frac{4}{\Omega}\omega\Big(\frac{\Olino }{\Omega}\Big)\otxb -\frac{2}{\Omega}\olinb \otx -
\frac{1}{2\Omega}\otxb ^2+2\Omega|\elin |^2\\
&\nonumber\quad+\Big(\frac{\Olino }{\Omega}\Big)\Big\{\frac{2}{\Omega}\partial_u\olin-\Omega\Big(\divs \divs \hat{\slashed{h}}-\frac{1}{2}\slashed{\Delta}\mathrm{tr}\slashed{h}\Big)-\frac{\mathrm{tr}\chi}{2}\Big(\accentset{(1)}{(\Omega\mathrm{tr}{\underline{\chi}})}-\otx \Big)\\
&\nonumber\quad+\frac{\mathrm{tr}\chi}{4}\Big(\Omega\mathrm{tr}\chi+4\omega\Big)\Big(\mathrm{tr}\slashed{h}-4\Big(\frac{\Olino }{\Omega}\Big)\Big)\Big\},
\end{align*}
where one uses that
\begin{align*}
   \Big\langle\elin +\eblin ,\divs \hat{\slashed{h}}-\frac{1}{2}\ns\mathrm{tr}\slashed{h} \Big\rangle\equiv -2\Big(\frac{\Olino }{\Omega}\Big)\Big[\divs \divs \hat{\slashed{h}}-\frac{1}{2}\slashed{\Delta}\mathrm{tr}\slashed{h}\Big].
\end{align*}
Using the linearised Gauss equation in proposition~\ref{LinGauss} gives the result stated in theorem~\ref{thm:CanEntoGusCon}.
\end{proof}
\subsection{Proof of Theorem~\ref{thm:CEntoGConb}}\label{betaconsproof}
\begin{proof}[Proof of theorem~\ref{thm:CEntoGConb}]
In this proof the following convention is adopted. The symbol $\equiv$ will denote equality under integration by parts of $\mathbb{S}^2_{u,v}$.

Now rather than go through the direct computation as in section~\ref{TMR} one can use the following idea to avoid (some of) the long computations. If $h$ satisfies the linearised null structure equations then $\mathcal{L}_{\Omega_k}h$ does. Note further that one can establish
\begin{align*}
(\mathcal{L}_{\Omega_k}h)_{33}=0=(\mathcal{L}_{\Omega_k}h)_{44}=(\mathcal{L}_{\Omega_k}h)_{3A}. 
\end{align*}
Therefore, if one expresses $(\mathcal{J}^T[h])^4-\frac{1}{r^2}\ns_3(r^2\mathcal{A}[h])$ and $(\mathcal{J}^T[h])^3+\frac{1}{r^2}\ns_4(r^2\mathcal{A}[h])$ in terms of $h$, then replacing with $\mathcal{L}_{{\Omega}_k}h$ will result in a simpler form for the components of $(\mathcal{J}^T[\mathcal{L}_{\Omega_k}h])$ (plus a boundary term). One can check that this yields 
\begin{align*}
(\mathcal{J}^T[\mathcal{L}_{\Omega_k}h])^4&=\frac{1}{\Omega}\Big(\Omega^2|\slashed{\mathcal{L}}_{{\Omega}_k}\xblin |^2+2\Omega^2|\slashed{\mathcal{L}}_{{\Omega}_k}\elin |^2-4\omega\Omega_k\Big(\frac{\Olino }{\Omega}\Big)\Omega_k(\otxb )\\
&\nonumber\quad-2\Omega_k(\olinb )\Omega_k(\otx )-
\frac{1}{2}\big(\Omega_k\otxb \big)^2\Big)+\frac{1}{r^2}\ns_3(r^2\mathcal{A}[\mathcal{L}_{{\Omega}_k}h]),\\
(\mathcal{J}^T[\mathcal{L}_{\Omega_k}h])^3&=\frac{1}{\Omega}\Big(2\Omega^2|\slashed{\mathcal{L}}_{{\Omega}_k}\eblin |^2+\Omega^2|\slashed{\mathcal{L}}_{{\Omega}_k}\xlin |^2+4\omega\Omega_k\Big(\frac{\Olino }{\Omega}\Big)\Omega_k\big(\otx \big)\Big)\\
&\nonumber\quad-2\Omega_k\big(\otxb \big)\Omega_k(\olin)-\frac{1}{2}\big(\Omega_k\otx \big)^2-\frac{1}{r^2}\ns_4(r^2\mathcal{A}[\mathcal{L}_{{\Omega}_k}h]).
\end{align*}

For simplicity, denote $\slashed{\mathcal{J}}^a\doteq \sum_k(\mathcal{J}^T[\mathcal{L}_{\Omega_k}h])^a$ and $\slashed{\mathcal{A}}\doteq \sum_{k}\mathcal{A}[\mathcal{L}_{{\Omega}_k}h]$.
Now, note that $\{\Omega_k\}_{k}$ satisfy the following identities
\begin{align}\label{metso3}
\sum_{k}\Omega_k^A\Omega_k^B&=r^2\slashed{g}^{AB}\qquad \Dts \Omega_k=0=\divs \Omega_k,
\end{align}
where the latter two identities follow from the Killing property. Next from lemma~\ref{Useful2DLemma}, a covector $\xi\in \Omega^1\mathbb{S}^2_{u,v}$ satisfies
\begin{align*}
|\slashed{\mathcal{L}}_{\Omega_k}\xi|^2&=|\ns_{\Omega_k}\xi|^2+\frac{1}{4}(\slashed{\mathrm{curl}}\Omega_k)^2|\xi|^2+(\slashed{\mathrm{curl}}\Omega_k)\slashed{\varepsilon}(\xi,\ns_{\Omega_k}\xi).
\end{align*}
One can check that $\sum_k(\slashed{\mathrm{curl}}\Omega_k)\Omega_k=0$. Using equation~(\ref{metso3}) and that $\sum_k(\slashed{\mathrm{curl}}\Omega_k)^2=4$ one finds
\begin{align*}
\sum_k|\slashed{\mathcal{L}}_{\Omega_k}\xi|^2&=r^2|\ns\xi|^2+|\xi|^2.
\end{align*}
Similarly, using lemma~\ref{Useful2DLemma}, for $\Theta\in \mathrm{symtr}(T\mathbb{S}^2_{u,v}\otimes T\mathbb{S}^2_{u,v})$ one has
\begin{align*}
\sum_{k}|\slashed{\mathcal{L}}_{\Omega_k}\Theta|^2=r^2|\ns\Theta|^2+4|\Theta|^2.
\end{align*}
Then, using the above results with $\xi\in \{\elin ,\eblin \}$ and $\Theta\in \{\xblin ,\xlin \}$ gives
\begin{align*}
\slashed{\mathcal{J}}^4&\equiv\frac{1}{\Omega}\Big(\Omega^2r^2|\ns\xblin |^2+4\Omega^2|\xblin |^2+2\Omega^2r^2|\ns\elin |^2+2\Omega^2|\elin |^2-
\frac{1}{2}r^2\big|\ns\otxb \big|^2\\
&\nonumber\quad-4r^2\omega\langle\ns\Big(\frac{\Olino }{\Omega}\Big),\ns(\otxb )\rangle-2r^2\langle\ns(\olinb ),\ns(\otx )\rangle\Big)+\frac{1}{r^2}\ns_3(r^2\slashed{\mathcal{A}}),\\
\slashed{\mathcal{J}}^3&\equiv\frac{1}{\Omega}\Big(2\Omega^2r^2|\ns\eblin |^2+2\Omega^2|\eblin |^2+\Omega^2r^2|\ns\xlin |^2+4\Omega^2|\xlin |^2-\frac{1}{2}r^2\big|\ns\otx \big|^2\\
&\nonumber\quad-2r^2\langle\ns\otxb ,\ns(\olin)\rangle+4r^2\omega\langle\ns\Big(\frac{\Olino }{\Omega}\Big),\ns\otx \rangle\Big)-\frac{1}{r^2}\ns_4(r^2\slashed{\mathcal{A}}).
\end{align*}
Recall that, from propositions~\ref{LinMet},~\ref{LinTorProp},~\ref{LinTorCon}, and lemma~\ref{Useful2DLemma}
\begin{equation}
    \begin{aligned}[c]
    \slashed{d}\Big(\frac{\Olino }{\Omega}\Big)&=\frac{1}{2}(\elin +\eblin ),\\
    \slashed{d}\olinb &=\frac{\Omega}{2}\Big(\ns_3\elin -\mathrm{tr}\chi \elin +\bblin \Big),\\
    |\ns\Theta|^2&\equiv 2|\divs \Theta|^2-\mathrm{Scal}(\slashed{g})|{\Theta}|^2,
    \end{aligned}
    \quad
    \begin{aligned}[c]
    \slashed{\mathrm{curl}}\elin &=\slin =-\slashed{\mathrm{curl}}\eblin ,\\
    \slashed{d}\olin&=\frac{\Omega}{2}\Big(\ns_4\eblin +\mathrm{tr}\chi \eblin -\blin \Big),\\
    |\ns\xi|^2&\equiv|\slashed{\mathrm{curl}}\xi|^2+|\divs \xi|^2-\frac{\mathrm{Scal}(\slashed{g})}{2}|\xi|^2.
    \end{aligned}
\end{equation}
Combining these results with $\xi\in \{\elin ,\eblin \}$ and $\Theta\in \{\xblin ,\xlin \}$ and the commutation lemma~\ref{comlem} gives
\begin{align*}
\slashed{\mathcal{J}}^4&\equiv\frac{1}{\Omega}\Big(2\Omega^2r^2|\divs \xblin |^2+2\Omega^2|\xblin |^2+2\Omega^2r^2|\divs \elin |^2+2\Omega^2r^2|\slin |^2-
\frac{1}{2}r^2\big|\ns\otxb \big|^2\\
&\nonumber\quad+2r^2\omega\divs (\elin +\eblin )\otxb +r^2\Omega\ns_3(\divs \elin )\otx -\frac{3}{2}r^2\Omega(\mathrm{tr}\chi)\divs (\elin )\otx \\
&\nonumber\quad+r^2\Omega\divs (\bblin )\otx \Big)+\frac{1}{r^2}\ns_3(r^2\slashed{\mathcal{A}}),\\
\slashed{\mathcal{J}}^3&\equiv\frac{1}{\Omega}\Big(2\Omega^2r^2(\divs \eblin )^2+2\Omega^2r^2|\slin |^2+\Omega^2r^2|\divs \xlin |^2+2\Omega^2|\xlin |^2-\frac{1}{2}r^2\big|\ns\otx \big|^2\\
&\nonumber\quad+r^2\Omega\otxb \ns_4(\divs \eblin )+\frac{3}{2}r^2\mathrm{tr}\chi\Omega\otxb \divs (\eblin )-r^2\Omega\otxb \divs (\blin )\\
&\nonumber\quad-2r^2\omega\divs (\elin +\eblin )\otx \Big)-\frac{1}{r^2}\ns_4(r^2\slashed{\mathcal{A}}).
\end{align*}
Using the linearised Codazzi equations in proposition~\ref{LinCod} one has
\begin{align*}
2\Omega^2r^2|\divs \xblin |^2&\equiv\frac{r^2}{2}|\ns\otxb |^2+\frac{\Omega^2r^2(\mathrm{tr}\chi)^2}{2}|\elin |^2+2\Omega^2r^2|\bblin |^2\\
&\nonumber\quad+2\Omega^2r^2(\mathrm{tr}\chi)\langle\elin ,\bblin \rangle-\Omega r^2\mathrm{tr}\chi\otxb \divs \elin -2\Omega r^2\otxb \divs \bblin ,\\
2\Omega^2r^2|\divs \xlin |^2&\equiv\frac{r^2}{2}|\ns\otx |^2+\frac{\Omega^2r^2(\mathrm{tr}\chi)^2}{2}|\eblin |^2+2\Omega^2r^2|\blin |^2\\
&\nonumber\quad+2\Omega^2r^2(\mathrm{tr}\chi)\langle\eblin ,\blin \rangle+\Omega r^2\mathrm{tr}\chi\otx \divs \eblin +2\Omega r^2\otx \divs \blin .
\end{align*}
Substituting into $\slashed{\mathcal{J}}^a$ gives
\begin{align*}
\slashed{\mathcal{J}}^4&\equiv\frac{1}{\Omega}\Big(2\Omega^2|\xblin |^2+2\Omega^2r^2|\divs \elin |^2+2\Omega^2r^2|\slin |^2+\frac{\Omega^2r^2}{2}(\mathrm{tr}\chi)^2|\elin |^2+2\Omega^2r^2|\bblin |^2+r^2\Omega\divs \bblin \otx \\
&\nonumber\quad-\frac{3\Omega\mathrm{tr}\chi}{2}r^2\divs \elin \otx +2r^2\omega\divs (\elin +\eblin )\otxb +r^2\Omega\ns_3(\divs \elin )\otx \\
&\nonumber\quad+2\Omega^2r^2(\mathrm{tr}\chi)\langle\elin ,\bblin \rangle-\Omega r^2\mathrm{tr}\chi\otxb \divs \elin -2\Omega r^2\otxb \divs \bblin \Big)+\frac{1}{r^2}\ns_3(r^2\slashed{\mathcal{A}}),\\
\slashed{\mathcal{J}}^3&\equiv\frac{1}{\Omega}\Big(2\Omega^2r^2(\divs \eblin )^2+2\Omega^2r^2|\slin |^2+2\Omega^2|\xlin |^2+\frac{\Omega^2r^2}{2}(\mathrm{tr}\chi)^2|\eblin |^2+2\Omega^2r^2|\blin |^2-r^2\Omega\otxb \divs \blin \\
&\nonumber\quad+r^2\Omega\otxb \ns_4(\divs \eblin )+\frac{3}{2}r^2\mathrm{tr}\chi\Omega\otxb \divs \eblin +2\Omega^2r^2(\mathrm{tr}\chi)\langle\eblin ,\blin \rangle+2\Omega r^2\otx \divs \blin \\
&\nonumber\quad+\Omega r^2\mathrm{tr}\chi\otx \divs \eblin -2r^2\omega\divs (\elin +\eblin )\otx \Big)-\frac{1}{r^2}\ns_4(r^2\slashed{\mathcal{A}}).\nonumber
\end{align*}
By considering similar ideas to the points (1) and (2) around equation~(\ref{Usefulremarkboundary}) in the proof of theorem~\ref{thm:CanEntoGusCon} in section~\ref{TMR} one can identify a boundary term. To this end, define
\begin{equation}
    \begin{aligned}[c]
    \mathcal{B}_1&\doteq r^4\Big(\otx \divs \elin -\otxb \divs \eblin \Big),\\
    \mathcal{B}_3&\doteq 2r^4\Omega\mathrm{tr}\chi\Big(\frac{\Olino }{\Omega}\Big)\accentset{(1)}{\rho,}\\
    \mathcal{B}_5&\doteq \frac{r^4}{2\Omega}\mathrm{tr}\chi\otxb \otx ,\\
    \end{aligned}\qquad 
    \begin{aligned}[c]
    \mathcal{B}_2&\doteq r^4\rlin \big(\otx -\accentset{(1)}{(\Omega\mathrm{tr}{\underline{\chi}})}\big),\\
    \mathcal{B}_4&\doteq r^4(\Omega\mathrm{tr}\chi)\langle\elin ,\eblin \rangle,\\
    \mathcal{B}_6&\doteq r^4\frac{\omega}{\Omega^2}\otx \otxb .
    \end{aligned}
\end{equation}
Computing $\ns_3\mathcal{B}_1$ and $\ns_4\mathcal{B}_1$ with propositions~\ref{LinRay}, \ref{LinTorProp} and \ref{LinPropExp} gives
\begin{align*}
\frac{1}{r^2}\ns_3\mathcal{B}_1&\equiv2r^2\mathrm{tr}\chi\olinb \divs \eblin -\frac{3}{2}r^2(\mathrm{tr}\chi)\otx \divs \elin +2r^2\Omega(\divs \elin )^2+2\Omega r^2\rlin \divs \elin -2\Omega r^2\rho|\elin |^2\\
&\nonumber\quad-2\Omega r^2\rho\langle\elin ,\eblin \rangle+r^2\otx \ns_3\divs \elin -r^2\otxb \divs \bblin +\frac{2r^2\omega}{\Omega}\otxb \divs \eblin \nonumber,
\end{align*}
and
\begin{align*}
\frac{1}{r^2}\ns_4\mathcal{B}_1&\equiv2r^2\mathrm{tr}\chi\olin\divs \elin -\frac{3}{2}r^2(\mathrm{tr}\chi)\otxb \divs \eblin -2r^2\Omega(\divs \eblin )^2-2\Omega r^2\rlin \divs \eblin +2\Omega r^2\rho|\underline{\elin }|^2\\
&\nonumber\quad+2\Omega r^2\rho\langle\elin ,\eblin \rangle-r^2\otxb \ns_4\divs \eblin -r^2\otx \divs \blin +\frac{2r^2\omega}{\Omega}\otx \divs \elin \nonumber.
\end{align*}
Again using propositions~\ref{LinRay},~\ref{LinPropExp} and the Bianchi equation for $\rlin $ in proposition~\ref{LinBianchi} one has that
\begin{align*}
\frac{1}{r^2}\ns_3\mathcal{B}_2&\equiv r^2\big(\otxb -\otx \big)\divs \bblin +\frac{3}{2}r^2\frac{\rho}{\Omega}\otxb ^2+2\Omega r^2|\rlin |^2\\
&\nonumber\quad-\frac{3}{2}r^2\frac{\rho}{\Omega}\otxb \otx 
+2\Omega r^2\rlin \divs \elin +4\Omega r^2\rho\Big(\frac{\Olino }{\Omega}\Big)\rlin \\
&\nonumber\quad-r^2\mathrm{tr}\chi\rlin \otxb +2r^2\frac{\omega}{\Omega}\rlin \otxb +2r^2\mathrm{tr}\chi\rlin \olinb ,\\
\frac{1}{r^2}\ns_4\mathcal{B}_2&\equiv r^2\big(\otx -\otxb \big)\divs \blin -\frac{3}{2}r^2\frac{\rho}{\Omega}\otx ^2-2\Omega r^2|\rlin |^2\\
&\nonumber\quad+\frac{3}{2}r^2\frac{\rho}{\Omega}\otxb \otx 
-2\Omega r^2\rlin \divs \eblin -4\Omega r^2\rho\Big(\frac{\Olino }{\Omega}\Big)\rlin \\
&\nonumber\quad-r^2\mathrm{tr}\chi\rlin \otx +2r^2\frac{\omega}{\Omega}\rlin \otx +2r^2\mathrm{tr}\chi\rlin \olin.
\end{align*}
Further, from propositions~\ref{LinMet} and the Bianchi equation for $\rlin $ in proposition~\ref{LinBianchi} one has
\begin{align*}
\frac{1}{r^2}\ns_3\mathcal{B}_3&\equiv2r^2(\mathrm{tr}\chi)\olinb \rlin +r^2(\Omega\mathrm{tr}\chi)\langle\bblin ,\elin \rangle+r^2(\Omega\mathrm{tr}\chi)\langle\bblin ,\eblin \rangle+4r^2\rho\Big(\frac{\Olino }{\Omega}\Big)\rlin \\
&\nonumber\quad-3r^2\rho(\mathrm{tr}\chi)\Big(\frac{\Olino }{\Omega}\Big)\otxb ,\\
\frac{1}{r^2}\ns_4\mathcal{B}_3&\equiv2r^2(\mathrm{tr}\chi)\olin\rlin -r^2(\Omega\mathrm{tr}\chi)\langle\blin ,\elin \rangle-r^2(\Omega\mathrm{tr}\chi)\langle\blin ,\eblin \rangle-4r^2\rho\Big(\frac{\Olino }{\Omega}\Big)\rlin \\
&\nonumber\quad-3r^2\rho(\mathrm{tr}\chi)\Big(\frac{\Olino }{\Omega}\Big)\otx .
\end{align*}
Noting that $r^2\rho\mathrm{tr}\chi=-4\Omega\omega=-\frac{4\omega}{\Omega}(1+r^2\rho)$ one can show, using proposition~\ref{LinTorProp} that
\begin{align*}
\frac{1}{r^2}\ns_3\mathcal{B}_4&\equiv-2r^2(\mathrm{tr}\chi)\olinb \divs \eblin +2r^2\rho\Omega\big(\langle\elin ,\eblin \rangle-|\elin |^2\big)-2\Omega|\elin |^2+r^2(\Omega\mathrm{tr}\chi)\langle\elin -\eblin ,\bblin \rangle,\\
\frac{1}{r^2}\ns_4\mathcal{B}_4&\equiv-2r^2(\mathrm{tr}\chi)\olin\divs \elin -2r^2\rho\Omega\langle\elin ,\eblin \rangle+2\Omega|\eblin |^2+2r^2\Omega\rho|\eblin |^2+r^2(\Omega\mathrm{tr}\chi)\langle\elin -\eblin ,\blin \rangle.
\end{align*}
Noting that $(\mathrm{tr}\chi)^2=\frac{4}{r^2}+4\rho$ and $\omega\mathrm{tr}\chi=-\Omega\rho$ with propositions~\ref{LinRay} and~\ref{LinPropExp} gives
\begin{align*}
\frac{1}{r^2}\ns_3\mathcal{B}_5&\equiv r^2(\mathrm{tr}\chi)\otxb (\divs \elin +\rlin )-8\frac{\omega}{\Omega}\Big(\frac{\Olino }{\Omega}\Big)\otxb -8\frac{\omega}{\Omega}\rho r^2\Big(\frac{\Olino }{\Omega}\Big)\otxb -\frac{1}{\Omega}\otxb ^2\\
&\nonumber\quad -\frac{r^2\rho}{\Omega}\otxb ^2+\frac{r^2\rho}{\Omega}\otxb \otx  -\frac{4}{\Omega}\olinb \otx -\frac{4r^2}{\Omega}\rho\olinb \otx ,\\
\frac{1}{r^2}\ns_4\mathcal{B}_5&\equiv r^2(\mathrm{tr}\chi)\otx (\divs \eblin +\rlin )-8\frac{\omega}{\Omega}\Big(\frac{\Olino }{\Omega}\Big)\otx -8\frac{\omega}{\Omega}\rho r^2\Big(\frac{\Olino }{\Omega}\Big)\otx +\frac{1}{\Omega}\otx ^2\\
&\nonumber\quad +\frac{r^2\rho}{\Omega}\otx ^2-\frac{r^2\rho}{\Omega}\otxb \otx +\frac{4}{\Omega}\olin\otxb +\frac{4r^2}{\Omega}\rho\olin\otxb .
\end{align*}
Finally, one can check that using propositions~\ref{LinRay} and~\ref{LinPropExp} that
\begin{align*}
\frac{1}{r^2}\ns_3\mathcal{B}_6&\equiv-\frac{r^2}{2\Omega}\rho\otx \otxb +\frac{2\omega r^2}{\Omega}\otxb \divs \elin +\frac{2r^2\omega}{\Omega}\rlin \otxb \\
&\quad+4r^2\frac{\omega\rho}{\Omega}\Big(\frac{\Olino }{\Omega}\Big)\otxb +\frac{1}{2}r^2\frac{\rho}{\Omega}\otxb ^2+\frac{2r^2\rho}{\Omega}\olinb \otx ,\nonumber\\
\frac{1}{r^2}\ns_4\mathcal{B}_6&\equiv\frac{r^2}{2\Omega}\rho\otx \otxb +\frac{2\omega r^2}{\Omega}\otx \divs \eblin +\frac{2r^2\omega}{\Omega}\rlin \otx \\
&\quad+4r^2\frac{\omega\rho}{\Omega}\Big(\frac{\Olino }{\Omega}\Big)\otx -\frac{1}{2}r^2\frac{\rho}{\Omega}\otx ^2-\frac{2r^2\rho}{\Omega}\olin\accentset{(1)}{(\Omega\mathrm{tr}{\underline{\chi}})}\nonumber.
\end{align*}
By writing that 
\begin{align*}
\mathcal{B}=\frac{1}{r^2}\big(\mathcal{B}_1-\mathcal{B}_2+\mathcal{B}_3+\mathcal{B}_4-\mathcal{B}_5+\mathcal{B}_6\big)+\sum_{k=1}^3\mathcal{A}[\mathcal{L}_{\Omega_k}h]
\end{align*}
one can now calculate $\sum(\mathcal{J}[\mathcal{L}_{\Omega_i}h])^4-\frac{1}{r^2}\ns_3(r^2\mathcal{B})$ and $\sum(\mathcal{J}[\mathcal{L}_{\Omega_i}h])^3+\frac{1}{r^2}\ns_4(r^2\mathcal{B})$ to show the desired result. 
\end{proof}
\subsection{Proof of Theorem~\ref{TH5}}\label{alphacons}
\begin{proof}[Proof of Theorem~\ref{TH5}]
In this proof the following convention is adopted. The symbol $\equiv$ will denote equality under integration by parts of $\mathbb{S}^2_{u,v}$.

There are two methods to verify the conservation law of theorem~\ref{TH5}. One can compute \underline{directly} using the linearised Bianchi identities (proposition~\ref{LinBianchi}) and linearised null structure equations (propositions~\ref{LinMet}-\ref{LinCod}) that the fluxes appearing in the statement are conserved. More precisely, one can compute that 
\begin{align*}
0&\equiv\frac{1}{r^2}\ns_3\Big[r^2\Big(\frac{\Omega^4}{4}|\alin |^2+\frac{3}{2}\Omega^4\big(|\rlin |^2+|\slin |^2+|\blin |^2\big)+\frac{\Omega^4}{2}|\bblin |^2+f_2|\xblin |^2+f_1|\xlin |^2+f_3|\eblin |^2-\frac{1}{\Omega^2}f_3\olin\otxb\\
&\nonumber\quad +\frac{2}{\Omega^2}\Big(\omega f_3+{2\Omega\mathrm{tr}\chi f_2}\Big)\Big(\frac{\Olino }{\Omega}\Big)\otx -\frac{f_1}{2\Omega^2}\otx ^2-\frac{f_2}{2\Omega^2}\Big[\otxb +2(\Omega\mathrm{tr}\chi)\Big(\frac{\Olino }{\Omega}\Big)\Big]^2\Big)\Big]\\
&\nonumber\quad+\frac{1}{r^2}\ns_4\Big[r^2\Big(\frac{\Omega^4}{4}|\ablin |^2+\frac{3}{2}\Omega^4\big(|\rlin |^2+|\slin |^2+|\bblin |^2\big)+\frac{\Omega^4}{2}|\blin |^2+f_1|\xblin |^2+f_2|\xlin |^2+ f_3|\elin |^2-\frac{f_3}{\Omega^2}\olinb \otx\\
&\nonumber\quad  -\frac{2}{\Omega^2}\Big(\omega f_3+{2\Omega\mathrm{tr}\chi f_2}\Big)\Big(\frac{\Olino }{\Omega}\Big)\otxb -\frac{f_1}{2\Omega^2}\otxb ^2-\frac{f_2}{2\Omega^2}\Big[\otx -2(\Omega\mathrm{tr}\chi)\Big(\frac{\Olino }{\Omega}\Big)\Big]^2\Big)\Big].
\end{align*}
Therefore, if one integrates over the spacetime region~$[u_0,u_1]\times[v_0,v_1]\times\mathbb{S}_{u,v}^2$ then one obtains the conservation law in the statement of theorem~\ref{TH5}. Given the fluxes, this is perhaps the simplest way to prove the conservation law.

The second way is more constructive proof and is completely analogous to the proof of theorem~\ref{thm:CEntoGConb} where one computes the canonical energy of $\mathcal{L}_Th$ and then manipulates the fluxes into a `satisfactory' form. This was how the author originally derived the conservation law. This is illustrated below, however, many explicit computations are left out for brevity but can be found in appendix~\cite{MyThesis}. 

Again rather than go through the direct computation of $(\mathcal{J}^T[\mathcal{L}_{T}h])^a$ as in section~\ref{TMR} one can use the same idea as in the proof of theorem~\ref{thm:CEntoGConb} to avoid (some of) the long computations. If $h$ satisfies the linearised null structure equations then $\mathcal{L}_{T}h$ does. Therefore, if one expresses $(\mathcal{J}^T[h])^4-\frac{1}{r^2}\ns_3(\mathcal{A}[h])$ and $(\mathcal{J}^T[h])^3+\frac{1}{r^2}\ns_4(\mathcal{A}[h])$ in terms of $h$, then replacing with $\mathcal{L}_{T}h$ will result in a simpler form for the components of $(\mathcal{J}^T[\mathcal{L}_{T}h])$ (plus a boundary term). One can check that this yields
\begin{align}
(\mathcal{J}^T[\mathcal{L}_{T}h])^4&\equiv\frac{1}{\Omega}\Big(\Omega^2|\slashed{\mathcal{L}}_{T}\xblin |^2+2\Omega^2|\slashed{\mathcal{L}}_{T}\elin |^2-4\omega T\Big(\frac{\Olino }{\Omega}\Big)T(\otxb )\label{jlthstart}\\
&\nonumber\qquad-2T(\olinb )T(\otx )-
\frac{1}{2}\big(T\otxb \big)^2\Big)+\frac{1}{r^2}\ns_3\mathcal{A}(\mathcal{L}_{T}h),\\
(\mathcal{J}^T[\mathcal{L}_{T}h])^3&\equiv\frac{1}{\Omega}\Big(2\Omega^2|\slashed{\mathcal{L}}_{T}\eblin |^2+\Omega^2|\slashed{\mathcal{L}}_{T}\xlin |^2+4\omega T\Big(\frac{\Olino }{\Omega}\Big)T\big(\otx \big)\\
&\nonumber\qquad-2T\big(\otxb \big)T(\olin)-\frac{1}{2}\big(T\otx \big)^2\Big)-\frac{1}{r^2}\ns_4\mathcal{A}(\mathcal{L}_{T}h).
\end{align}
Calculating using the proposition~\ref{LinShear} gives:
\begin{align*}
\slashed{\mathcal{L}}_T\xblin &=-\frac{\Omega}{2}\ablin -\omega\xblin -\Omega\Dts \eblin +\frac{1}{4}(\Omega\mathrm{tr}\chi)\Big(\xblin +\xlin \Big),\\
\slashed{\mathcal{L}}_T\xlin &=-\frac{\Omega}{2}\alin +\omega\xlin -\Omega\Dts \elin -\frac{1}{4}(\Omega\mathrm{tr}\chi)\Big(\xblin +\xlin \Big).
\end{align*}
Further, from proposition~\ref{LinTorProp}
\begin{align*}
\slashed{\mathcal{L}}_T\elin &=\ns\olinb +\frac{1}{4}\Omega\mathrm{tr}\chi(\elin +\eblin )-\frac{\Omega}{2}\Big(\blin +\bblin \Big),\\
\slashed{\mathcal{L}}_T\eblin &=\ns\olin-\frac{1}{4}\Omega\mathrm{tr}\chi(\elin +\eblin )+\frac{\Omega}{2}\Big(\blin +\bblin \Big).
\end{align*}
One can additionally compute $T(\otxb )$, $T(\otx )$, $T(\olin)$ and $T(\olinb )$ from propositions~\ref{LinRay}, \ref{LinPropExp} and~\ref{Linomega}. Note that under integration by parts on the spheres and using the torsion equations in proposition~\ref{LinTorProp}:
\begin{align}
|\Dts \eblin |^2\equiv\frac{1}{2}|\slin |^2-\Big[\frac{1}{4}(\Omega\mathrm{tr}\chi)+\omega\Big]\frac{\mathrm{tr}\chi}{\Omega}|\eblin |^2+\frac{1}{2}|\divs \eblin |^2,\label{jlthfinal}
\end{align}
and analogously for $|\Dts \elin |^2$. This allows one to compute two (rather horrendous) expressions for $(\mathcal{J}^T[\mathcal{L}_Th])^4-\frac{1}{r^2}\ns_3(\mathcal{A}[\mathcal{L}_Th])$ and $(\mathcal{J}^T[\mathcal{L}_Th])^3+\frac{1}{r^2}\ns_4\mathcal{A}[\mathcal{L}_Th]$ where $\mathcal{A}[\mathcal{L}_{T}h]$ results from expressing $\mathcal{A}[h]$ above in terms of $h$ and replacing it with $\mathcal{L}_{T}h$. 

At this point one can then (arduously) identify the following boundary term to use in the manipulation of the resulting flux densities:
\begin{align*}
\mathcal{C}[h]&\doteq \frac{\Omega^2}{2}\Big(\frac{1}{4}(\Omega\mathrm{tr}\chi)-\omega\Big)\Big[|\xlin |^2+|\xblin |^2\Big]+\frac{\Omega^2}{4}(\Omega\mathrm{tr}\chi)\langle\xblin ,\xlin \rangle-\Omega^3\Big[\langle\bblin ,\eblin \rangle+\langle\blin ,\elin \rangle\Big]\\
&\nonumber\quad+\olin T(\otxb )-\olinb T(\otx )+(\Omega\mathrm{tr}\chi)\olin\olinb -2\Omega^2\omega\Big(\frac{\Olino }{\Omega}\Big)\rlin +\Omega^2\omega\langle\elin ,\eblin \rangle\\
&\nonumber\quad+\frac{1}{4}\Omega^2(\Omega\mathrm{tr}\chi)\Big[|\elin |^2+|\eblin |^2\Big]+\frac{1}{4}\Big(\omega-\frac{1}{4}(\Omega\mathrm{tr}\chi)\Big)\Big[\otxb ^2+\otx ^2\Big]\\
&\nonumber\quad-\frac{1}{8}(\Omega\mathrm{tr}\chi)\otx \otxb +\Big(2\omega^2-\frac{3}{2}\Omega^2\rho\Big)\Big(\frac{\Olino }{\Omega}\Big)\Big[\otxb -\otx \Big]\\
&\nonumber\quad-\Big(4\Omega^2\omega\rho+\frac{3}{2}\Omega^2(\Omega\mathrm{tr}\chi)\rho\Big)\Big(\frac{\Olino }{\Omega}\Big)^2+\sum_k\mathcal{A}[\mathcal{L}_{T}h].
\end{align*}
Therefore, computing $(\mathcal{J}^T[\mathcal{L}_Th])^4-\frac{1}{r^2}\ns_3(r^2\mathcal{C})$ and $(\mathcal{J}^T[\mathcal{L}_Th])^3+\frac{1}{r^2}\ns_4(r^2\mathcal{C})$ allows one to show $T$-canonical energy of $\mathcal{L}_{T}h$ satisfies
\begin{align*}
{\mathcal{E}}^T_u[\mathcal{L}_{T}h](v_0,v_1)&=2\dot{\overline{{\mathcal{E}}}}^T_u[h](v_0,v_1)-2\int_{\mathbb{S}_{u,v}^2}\mathcal{C}[h](u,v,\theta,\phi)\slashed{\varepsilon}\Big|_{v_0}^{v_1},\\
{\mathcal{E}}^T_v[\mathcal{L}_{T}h](u_0,u_1)&=2\dot{\overline{{\mathcal{E}}}}^T_v[h](u_0,u_1)+2\int_{\mathbb{S}_{u,v}^2}\mathcal{C}[h](u,v,\theta,\phi)\slashed{\varepsilon}\Big|_{u_0}^{u_1}.
\end{align*}
By the conservation of the canonical energy and the cancellation of the boundary terms on the spheres $\mathbb{S}_{u,v}^2$ gives the result stated in theorem~\ref{TH5}.
\end{proof}
\pagebreak
\footnotesize
\bibliographystyle{IEEEtranS}
\bibliography{CL1bib}

\begin{thebibliography}{10}
\providecommand{\url}[1]{#1}
\csname url@samestyle\endcsname
\providecommand{\newblock}{\relax}
\providecommand{\bibinfo}[2]{#2}
\providecommand{\BIBentrySTDinterwordspacing}{\spaceskip=0pt\relax}
\providecommand{\BIBentryALTinterwordstretchfactor}{4}
\providecommand{\BIBentryALTinterwordspacing}{\spaceskip=\fontdimen2\font plus
\BIBentryALTinterwordstretchfactor\fontdimen3\font minus
  \fontdimen4\font\relax}
\providecommand{\BIBforeignlanguage}[2]{{%
\expandafter\ifx\csname l@#1\endcsname\relax
\typeout{** WARNING: IEEEtranS.bst: No hyphenation pattern has been}%
\typeout{** loaded for the language `#1'. Using the pattern for}%
\typeout{** the default language instead.}%
\else
\language=\csname l@#1\endcsname
\fi
#2}}
\providecommand{\BIBdecl}{\relax}
\BIBdecl

\bibitem{Alinhac}
S.~Alinhac, \emph{{Geometric Analysis of Hyperbolic Differential Equations: An
  Introduction}}.\hskip 1em plus 0.5em minus 0.4em\relax Cambridge University
  Press, 2010.

\bibitem{AndersonBlue}
L.~Andersson and P.~Blue, ``Hidden symmetries and decay for the wave equation
  on the {Kerr} spacetime,'' \emph{Ann. Math.}, vol. 182, no.~3, pp. 787--853,
  2015.

\bibitem{ABBSM19}
L.~Andersson, T.~Bäckdahl, P.~Blue, and S.~Ma, ``{Stability for linearized
  gravity on the Kerr spacetime},'' 2019, arxiv:1903.03859.

\bibitem{Aretakis}
S.~Aretakis, ``Lecture {Notes} {General} {Relativity} {Columbia}
  {University},'' 2012.

\bibitem{Chandrasekhar2}
S.~Chandrasekhar, \emph{The {Mathematical} {Theory} of {Black} {Holes}}.\hskip
  1em plus 0.5em minus 0.4em\relax Oxford Univ. Press, New York, 1992.

\bibitem{Chandrasekhar3}
S.~Chandrasekhar and V.~Ferrari, ``{The Flux Integral for Axisymmetric
  Perturbations of Static Space-Times},'' \emph{Proc. Roy. Soc. Lon. A}, vol.
  428, no. 1875, pp. 325--349, 1990.

\bibitem{FormBH}
D.~Christodoulou, \emph{The {Formation} of {Black} {Holes} in {General}
  {Relativity}}.\hskip 1em plus 0.5em minus 0.4em\relax EMS, 2009.

\bibitem{MyThesis}
S.~C. Collingbourne, ``{The Gregory--Laflamme Instability and Conservation Laws
  for Linearised Gravity (Doctoral thesis)},'' \emph{University of Cambridge},
  2022.

\bibitem{HolzegelSC}
S.~C. Collingbourne and G.~Holzegel, ``{Uniform Boundedness for Solutions to
  the Teukolsky Equation on Schwarzschild from Conservation Laws of Linearised
  Gravity},'' 2023, to Appear.

\bibitem{DR}
M.~Dafermos and I.~Rodnainski, ``Lectures on {Black} {Holes} and {Linear}
  {Waves},'' \emph{Clay Math. Proc.}, vol.~17, pp. 97--206, 2008.

\bibitem{DHR2}
M.~Dafermos, G.~Holzegel, and I.~Rodnianski, ``{Boundedness and Decay for the
  Teukolsky Equation on Kerr Spacetimes {I}: The Case $|a|<<{M}$},'' \emph{Ann.
  PDE}, vol.~5, no.~1, p.~2, 2019.

\bibitem{DHR}
------, ``The linear stability of the {Schwarzschild} solution to gravitational
  perturbations,'' \emph{Acta Math.}, vol. 222, no.~1, pp. 1--214, 2019.

\bibitem{DHRT}
M.~Dafermos, G.~Holzegel, I.~Rodnianski, and M.~Taylor, ``{The non-linear
  stability of the Schwarzschild family of black holes},'' 2021,
  arXiv:2104.08222.

\bibitem{DR2}
M.~Dafermos and I.~Rodnianski, ``A proof of the uniform boundedness of
  solutions to the wave equation on slowly rotating {Kerr} backgrounds,''
  \emph{Invent. Math.}, vol. 185, pp. 467--559, 2011.

\bibitem{Mihalis2}
------, ``A new physical-space approach to decay for the wave equation with
  applications to black hole spacetimes,'' \emph{XVIth Int. Con. Math. Phys.},
  2009.

\bibitem{DR4}
------, ``The red-shift effect and radiation decay on black hole spacetimes,''
  \emph{Commun. Pure Appl. Math.}, vol.~62, no.~7, pp. 859--919, 2009.

\bibitem{DR3}
------, ``{Decay for solutions of the wave equation on Kerr exterior spacetimes
  I-II: The cases $|a| << M$ or axisymmetry},'' 2010, arXiv:1010.5132.

\bibitem{Mihalis}
M.~Dafermos, I.~Rodnianski, and Y.~Shlapentokh-Rothman, ``Decay for solutions
  of the wave equation on {Kerr} exterior spacetimes {III}: The full
  subextremal case $|a| < {M}$,'' 2014, arXiv:1402.7034.

\bibitem{DRY}
------, ``\BIBforeignlanguage{English (US)}{{{Decay for solutions of the wave
  equation on {Kerr} exterior spacetimes III: The full subextremal case $|a| <
  M$}}},'' \emph{\BIBforeignlanguage{English (US)}{Ann. Math.}}, vol. 183,
  no.~3, pp. 787--913, 2016.

\bibitem{Friedman78}
J.~L. Friedman, ``{Generic instability of rotating relativistic stars},''
  \emph{{Comm. Math. Phys.}}, vol.~62, no.~3, pp. 247--278, 1978.

\bibitem{Elena}
E.~Giorgi, ``{Boundedness and Decay for the Teukolsky System of Spin $\pm 2$ on
  Reissner--Nordström Spacetime: The Case $Q<< M$},'' \emph{Ann. Hen. Poin.},
  vol.~21, no.~8, pp. 2485--2580, 2020.

\bibitem{GKS}
E.~Giorgi, S.~Klainerman, and J.~Szeftel, ``{Wave equations estimates and the
  nonlinear stability of slowly rotating Kerr black holes},'' 2022,
  arXiv:2205.14808.

\bibitem{HolW}
S.~Hollands and R.~M. Wald, ``Stability of {Black} {Holes} and {Black}
  {Branes},'' \emph{Commun. Math. Phys.}, vol. 321, no.~3, p. 629, 2013.

\bibitem{Holzegel}
G.~Holzegel, ``Conservation laws and flux bounds for gravitational
  perturbations of the {Schwarzschild} metric,'' \emph{Class. Quant. Grav.},
  vol.~33, no.~20, p. 205004, 2016.

\bibitem{Joe}
J.~Keir, ``Stability, instability, canonical energy and charged black holes,''
  \emph{Class. Quant. Grav.}, vol.~31, no.~3, p. 035014, 2014.

\bibitem{Szeftel}
S.~Klainerman and J.~Szeftel, \emph{{Global Nonlinear Stability of
  Schwarzschild Spacetime Under Polarized Perturbations}}, ser. Ann. Math.
  Stud.\hskip 1em plus 0.5em minus 0.4em\relax Princeton University Press,
  2020.

\bibitem{KS3}
------, ``{Kerr stability for small angular momentum},'' 2021,
  arXiv:2104.11857.

\bibitem{Moncrief2}
V.~Moncrief and N.~Gudapati, ``{A Positive-Definite Energy Functional for the
  Axisymmetric Perturbations of Kerr-Newman Black Holes},'' 2021,
  arXiv:2105.12632.

\bibitem{Morawetz}
C.~S. Morawetz, ``Decay for solutions of the exterior problem for the wave
  equation,'' \emph{Commun. Pure Appl. Math.}, vol.~28, no.~2, pp. 229--264,
  1975.

\bibitem{PW2}
K.~Prabhu and R.~M. Wald, ``Canonical energy and {Hertz} potentials for
  perturbations of {Schwarzschild} spacetime,'' \emph{Class. Quant. Grav.},
  vol.~35, no.~23, p. 235004, 2018.

\bibitem{YakovRita1}
Y.~Shlapentokh-Rothman and R.~Teixeira~da Costa, ``{Boundedness and decay for
  the {Teukolsky} equation on {Kerr} in the full subextremal range $|a|<M$:
  frequency space analysis},'' 2020, arXiv:2007.07211.

\bibitem{YakovRita2}
------, ``{Boundedness and decay for the {Teukolsky} equation on {Kerr} in the
  full subextremal range $|a|<M$: physical space analysis},'' 2023,
  arXiv:2302.08916.

\end{thebibliography}
\end{document}